\documentclass[11pt]{article}
\usepackage{graphicx,amsmath,amsfonts,amssymb,amsthm,fullpage,times,color}
\catcode`\@ 11
\catcode`\@ 12
\numberwithin{equation}{section}

\def\sf#1#2{{\textstyle\frac{#1}{#2}}} 
\newcommand{\myatop}[2]{\genfrac{}{}{0pt}{}{#1}{#2}}  
\def\Re{\mathop{\rm Re}\nolimits}
\def\Im{\mathop{\rm Im}\nolimits}

\allowdisplaybreaks

\theoremstyle{plain}  
\newtheorem{theorem}{Theorem}[]
\newtheorem{lemma}[theorem]{Lemma}

\begin{document}

\title{Dynamics of KPI lumps}
\author{Sarbarish Chakravarty and Michael Zowada \\[1ex]
\small\it\
Department of Mathematics, University of Colorado, Colorado Springs, CO 80918 \\}
\date{}
\maketitle \kern-2\bigskipamount

\begin{abstract}
A family of nonsingular rational solutions of the 
Kadomtsev-Petviashvili (KP) I equation are investigated.
These solutions have multiple peaks whose heights are 
time-dependent and the peak trajectories in the
$xy$-plane are altered after collision. Thus
they differ from the standard multi-peaked KPI simple $n$-lump solutions 
whose peak heights as well as peak trajectories remain unchanged after 
interaction.The anomalous scattering occurs due to a non-trivial
internal dynamics among the peaks in a slow time scale. This phenomena
is explained by relating the peak locations to the roots of
complex heat polynomials. It follows from
the long time asymptotics of the solutions that the peak trajectories
separate as $O(\sqrt{|t|})$ as $|t| \to \infty$, and all the peak heights
approach the same constant value corresponding to that of the simple 
1-lump solution. Consequently, a multi-peaked $n$-lump solution 
evolves to a superposition of $n$ 1-lump solutions asymptotically 
as $|t| \to \infty$.
 
\end{abstract}

\thispagestyle{empty}

\section{Introduction}

Many nonlinear wave equations admit
special classes of exact solutions in the form of solitary waves 
that are of physical interest. Such equations arise in such diverse 
fields as fluid dynamics, nonlinear optics, magnetic systems, and plasma physics. 
An important example of such special nonlinear wave equations  
is the Kadomtsev-Petviashvili (KP) equation, which is a $(2+1)$-dimensional 
dispersive equation describing the propagation of small amplitude, 
long wavelength, uni-directional waves with small transverse variation 
(i.e., quasi-two-dimensional waves). It was originally proposed by Kadomtsev 
and Petviashvili~\cite{KP70} and has found applications in the study
of ion-acoustic (see e.g.~\cite{IR00,L98} and references therein),
and shallow water waves (see e.g. the monographs~\cite{AS81,A11,K18}).
The KP equation is also an integrable nonlinear equation with remarkably 
rich mathematical structure which are documented
in many research monographs (see e.g.~\cite{AS81,NMPZ1984,AC91,IR00,H04,K18}).

There are two mathematically distinct versions of the KP equation, referred
to as KPI and KPII. This article is concerned with the KPI equation which
can be expressed as
\begin{equation}\label{kp}
(4u_t+6uu_x+u_{xxx})_x=3u_{yy}.
\end{equation}
Here $u=u(x,y,t)$ represents the normalized wave amplitude at the point $(x,y)$ in
the $xy$-plane for fixed time $t$, and the subscripts denote partial derivatives.
Switching $3u_{yy}$ to $-3u_{yy}$ in \eqref{kp} yields the KPII equation. 
From the water wave theory perspective, KPI corresponds to large surface tension
while KPII arises in the small surface tension limit of the 
multiple-scale asymptotics~\cite{AS81,A11}.

The KPI equation admits large classes of exact rational solutions
known as lumps which are localized in the $xy$-plane and are non-singular
for all $t$. The simplest type of rational solutions was first discovered 
analytically by employing the dressing method~\cite{MZBIM77} and subsequently 
via the Hirota method~\cite{SA79}. 
These simple $n$-lump solutions consist of $n$ peaks 
which interact without any change of form or phase. Each peak
travels with distinct asymptotic velocity and
its trajectory remains unchanged before and after interaction 
as $|t| \to \infty$. These solutions were also obtained via the inverse
scattering transform (IST) in~\cite{FA83} where it was shown that the 
simple lumps (for a fixed $t$) correspond to bound state potentials
associated with complex conjugate pairs of simple eigenvalues of the 
non-stationary Schr\"odinger equation. Yet another class of KPI rational
solutions arise as bound states corresponding to eigenvalues with multiplicities
associated with the time-dependent Schr\"odinger equation, and are also
amenable to IST methodology~\cite{AV97,VA99}. These are called the multi-lump 
solutions which were originally found in~\cite{JT78} by algebraic
techniques and further investigated by several 
authors~\cite{PS93,GPS93,ACTV00}. The multi-lump solution is
an ensemble of a finite number of localized structures (or peaks) 
interacting in a non-trivial manner
unlike the $n$-simple lump solution. The peaks move with the same center-of-mass
velocity but undergo anomalous scattering with a non-zero deflection angle
after collision. Furthermore, the peak amplitudes evolve in time and reaches
a constant asymptotic value which equals that of the simple $1$-lump peak.
It has been also known~\cite{K78,GPS93,S94} that the dynamics of the KPI lumps are 
related to the multi-particle Calogero-Moser system; this connection was further 
explored in~\cite{P94,P98}.  

The purpose of this article is to revisit a special family of
multi-lump solutions referred throughout the article as the
$n$-lump solutions which are relatively simple to construct.
We carry out a detailed study of their dynamics, which, to the
best of our knowledge, was not done earlier.
The family of solutions considered here exhibit a simple yet 
interesting interaction pattern, namely, that in a comoving frame 
the multiple lumps interact 
along a line and move away from each other along an orthogonal
line in the plane after the interaction process is complete. 
Such anomalous scattering process is different from the usual
scattering of KP solitons or even the simple $n$-lump
solutions. The comoving frame travels with a uniform velocity
in the $xy$-plane, while the internal dynamics of the component
lumps takes place at a slower time that scales as $O(|t|^{1/2})$.
The anomalous scattering is explained by analyzing the
explicit formulas for the simpler solutions and then via
asymptotic analysis for the general case. The 
long-time asymptotics of the solutions are also worked out
in order to demonstrate that indeed the
solution splits up into $n$ distinct peaks like it is usually 
assumed in the literature. The peak locations evolve as a dynamical
system which is a reduction of the Calogero-Moser system, and can
be related to the zeros of polynomial solutions of the complex
two-dimensional heat equation. An elementary analysis of the
dynamical system provides a simple analytical explanation of
the anomalous (90-degree) scattering of the $n$-lump solutions.
Finally, the eigenfunction of the KP Lax pair associated with the
$n$-lump solutions are derived via binary Darboux transformations
to make connection with the results obtained by the IST method.
The binary Darboux transformation yields simple explicit formula 
for the $n$-lump eigenfunction that is rather difficult to obtain
from the IST method. Although some topics of this article had
been investigated in the past by several authors (including one
of the present authors), we believe that the explicit formulas 
for the polynomial $\tau$-function and the $n$-lump eigenfunction,
along with the long time asymptotics developed to study the peak dynamics
in terms of heat polynomials, are some of the new results obtained
in this paper.

The paper is structured as follows: in Section 2 the
generalized Schur polynomials are introduced; these 
polynomials play a fundamental role in the
theory of rational, multi-lump solutions of KPI. Then a 
polynomial form of the $n$-lump $\tau$-function is constructed
in terms of these polynomials. Some examples
of solutions for $n = 1,2,3$ are discussed in Section 3. Section 4 
is devoted to the asymptotic analysis of the $n$-lump solution for 
large $|t|$ and the connection to the complex heat polynomials
in two-dimension. The binary Darboux transformation and its 
application to construct the eigenfunction for the Lax pair
associated with the $n$-lump solution and connection with some
results obtained via IST are presented in Section 5. Section 6
contains some concluding remarks including future directions
of this continued investigation of KPI lumps. 
Finally, an appendix validating the approximation used throughout
Sections 3 and 4 to estimate the exact locations of the $n$-lump peaks, 
is included.
\section{Construction of special multi-lumps}
In this section we describe how to construct a special class of the
multi-lump solutions in a straight-forward, algebraic fashion. Although 
this construction
is related to the binary Darboux transformation or the so called Grammian
technique, no prior knowledge of such methods is required for this particular case.
The building blocks for these solutions are given by a special type of
complex polynomials called the generalized Schur polynomials which are introduced
below.

\subsection{Generalized Schur polynomials}
Let $k$ be a complex parameter and $\theta := kx+k^2y+k^3t+\gamma(k)$ where 
$(x,y,t) \in \mathbb{R}^3$ and $\gamma(k)$ is an arbitrary function
which is differentiable to all orders. Then the generalized Schur polynomial $p_n$
is defined via the relation
\begin{equation}
\phi_n:=\frac{1}{n!}\partial_k^n\exp(i\theta) = p_n\exp(i\theta)\,.
\label{sp}
\end{equation}
It is easy to see that $p_n = p_n(\theta_1, \theta_2, \cdots, \theta_n)$
where $\theta_j := i\partial_k^j\theta/j!$. However, only the first three
\begin{equation}
\theta_1 = i(x+2ky+3k^2t+\gamma_1), \quad \theta_2 =i(y+3kt+\gamma_2),
\quad \theta_3 = i(t+\gamma_3)
\label{theta}
\end{equation}
depend on $(x,y,t)$ while $\theta_j = i\gamma_j$
for $j>3$. Here, $\gamma_j(k) = \partial_k^j\gamma(k)/j!, \, j=1,\cdots,n$ are
to be viewed as independent complex parameters depending on the $k$-derivatives
of $\gamma(k)$. A generating function for the $p_n$'s is given by the Taylor series
\begin{equation}
\exp(i\theta(k+h)) = \exp(i\theta) \sum_{n=0}^\infty p_n(k)h^n \,,  
\label{gen}
\end{equation}
which yields, after expanding $i\theta(k+h)=i\theta(k)+h\theta_1+h^2\theta_2+\cdots$,
and comparing with the right hand side, an explicit expression for $p_n$, namely
\begin{equation}
p_n = \hspace{-.2in} \sum_{\myatop{m_1,m_2,\cdots,m_n \geq 0}{m_1+2m_2+\cdots nm_n=n}}  
\prod_{j=1}^n\frac{\theta_j^{m_j}}{m_j!} \,.
\label{pn}
\end{equation}   
The first few generalized Schur polynomials are given by
\[p_0=1, \quad p_1=\theta_1, \quad p_2 = \sf12\theta_1^2+\theta_2 \quad
p_3 = \sf{1}{3!}\theta_1^3+\theta_1\theta_2+\theta_3 \ldots \,.\]
It follows from \eqref{pn} that the $p_n$ is a {\it weighted}\, homogeneous
polynomial of degree $n$ in $\theta_j, \, j=1,\ldots,n$, i.e.,
$p_n(a\theta_1,a^2\theta_2,\cdots,a^n\theta_n)=
a^np_n(\theta_1,\theta_2,\cdots,\theta_n)$, where weight$(\theta_j)=j$.
Some useful properties including recurrence relations for the $p_n$'s can
be derived using \eqref{sp} and \eqref{gen}. These are listed below and will be
used throughout this article. 
\begin{subequations}
\begin{equation}
\partial^j_{\theta_1}p_n = \partial_{\theta_j} p_n = 
\left\{ \begin{matrix} p_{n-j}\,, & \quad j\leq n \cr
0\,, & \quad j>n \end{matrix} \right.  \label{propa} 
\end{equation}
\begin{equation}
p_{n+1}(k) = \sf{1}{n+1}\sum_{j=0}^n (j+1)\,\theta_{j+1}\,p_{n-j}\,, 
\quad n \geq 0, \qquad p_0=1 
\label{propb}
\end{equation} 
\begin{equation}
p_n(\theta_1+h_1, \theta_2+h_2, \cdots, \theta_n+h_n) =
\sum_{j=0}^n p_j(h_1,h_2,\cdots,h_j) \, p_{n-j}(\theta_1,\theta_2,\cdots,\theta_{n-j})
\label{propc}
\end{equation}
\end{subequations}
\paragraph{Remarks}
\begin{itemize}
\item[(a)] The generalized Schur polynomials were used earlier in
the study of rational solutions for the Zakharov-Shabat and KP 
hierarchies~\cite{M79,P94,P98}. These involve
multi-time variables with the phase defined as quasi-polynomial
$\theta = kt_1+k^2t_2+k^3t_3+\cdots$. In this paper, we restrict the $p_n$'s
to depend only on the first three variables $(t_1,t_2,t_3)=(x,y,t)$ while the
dependence on the remaining variables are parametric, through the complex
parameters $\theta_j= i\gamma_j$ for $j>3$.
\item[(b)] By choosing $\theta$ instead of $i\theta$ where
$\theta = kt_1+k^2t_2+k^3t_3+\cdots$ depends on infinitely many ``time'' variables,
and by evaluating the $k$-derivatives at $k=0$ in \eqref{sp}, one
recovers the standard Schur polynomials which play an important role in the
Sato theory of the KP hierarchy~\cite{S81} (see also~\cite{OSTT88,K18}). 
In fact, then the Schur polynomials are generated  
by \eqref{gen} with $k=0$ (and $i\theta \to \theta$), and are given explicitly
by the formula \eqref{pn} with $\theta_j \to t_j$.
\end{itemize} 
\subsection{The multi-lump $\tau$-function}
The solution of the KPI equation \eqref{kp} can be expressed as
\begin{equation}
u(x,y,t) = 2 (\ln \tau)_{xx} \,,
\label{tau}
\end{equation}
where the function $\tau(x,y,t)$ is known as the the 
$\tau$-function~\cite{S81,H04}. Here we describe how to construct
a polynomial $\tau$-function which yields a special family of
rational multi-lump solutions of KPI.

First, we note that the exponential $\exp(i\theta)$ introduced in Section 2.1
as well as its $k$-derivative of any order, solve the linear system
\begin{equation}
i\phi_y = \phi_{xx}, \qquad \quad \phi_t+\phi_{xxx} =0 \,,
\label{phi}
\end{equation}
where $\phi(x,y,t,k)$ is parametrized by $k \in \mathbb{C}$.
Next we take $\phi_n(x,y,t,k)=p_n\exp(i\theta)$ as given in \eqref{sp}  
and its complex conjugate $\bar{\phi}_n$ to define a $n$-lump $\tau$-function
as follows:
\begin{equation}
\tau_n(x,y,t) = \int_x^\infty \phi_n\bar{\phi}_n\,dx' =
\int_x^\infty|p_n|^2\exp i(\theta-\bar{\theta})\,dx' \,,
\label{taun}
\end{equation}
where the parameter $k$ in $\theta=kx+k^2y+k^3t+\theta_0(k)$ is chosen
such that $b:=\Im(k)>0$ in order for the integral above to converge. The
integral in \eqref{taun} may be evaluated by integration by parts, then 
from \eqref{tau}, the $n$-lump solution is given by 
\begin{equation}
u_n(x,y,t) = 2 (\ln \tau_n)_{xx} = 2 (\ln F_n)_{xx}, \quad \qquad
F_n = \sum_{j=0}^{2n}\frac{\partial^j_x|p_n|^2}{(2b)^j} \,,
\label{Fn}
\end{equation}
where the last equality follows since the factor 
$\sf{1}{2b}\exp i(\theta-\bar{\theta})$
in $\tau_n$ is anihilated by $\ln(\cdot)_{xx}$. Recall from Section 2.1 that 
$p_n$ is a polynomial of degree $n$ in $x,y,t$, hence $F_n$ is a polynomial
of degree $2n$ in $x,y,t$. Consequently, $u_n(x,y,t)$ is a rational function
of its arguments and decays as $\sf{1}{R^2}$ as $R=\sqrt{x^2+y^2} \to \infty$
for fixed $t$.

It can be directly shown that the $\tau_n$ in \eqref{taun} is indeed
a $\tau$-function for the KPI equation. Inserting \eqref{tau} into \eqref{kp},
integrating the resulting equation twice with respect to $x$, 
and setting the integration
constants to zero (since for lumps, the derivatives of $\ln(\tau)$ 
vanish as $R \to \infty$), one obtains 
\[ 4(\tau\tau_{xt}-\tau_x\tau_t)-3\tau_{xx}^2-\tau\tau_{xxxx}+
4\tau_x\tau_{xxx}-3\tau_y^2+3\tau\tau_{yy}=0 \,.\]
Next, by setting $\tau=\tau_n$ from \eqref{taun} in the above equation,
and using \eqref{phi} and its complex conjugate with $\phi=\phi_n$, 
a straight-forward calculation yields the desired result.

It is evident from \eqref{Fn} that the multi-lump solution is in effect
determined by the polynomial part of of $\tau_n$, namely, $F_n(x,y,t)$.
The expression for $F_n$ in \eqref{Fn} can be cast as a Hermitian form
as follows:
\begin{align*}
F_n(x,y,t) &= \sum_{j=0}^{2n}\frac{1}{(2b)^j} 
\sum_{l=0}^j\binom{j}{l}\partial_x^lp_n\partial_x^{j-l}\bar{p}_n 
= \sum_{l=0}^n\sum_{m=0}^n
\frac{1}{(2b)^{l+m}}\binom{l+m}{m}\partial_x^lp_n\partial_x^m\bar{p}_n \\
&=p^{\dagger}Cp\,, \qquad \text{with} \quad 
p=(p_n, \partial_xp_n, \cdots, \partial_x^np_n), 
\qquad C_{lm} = \frac{1}{(2b)^{l+m}}\binom{l+m}{m} \,,
\end{align*}
and where $C$ is a real, symmetric $(n+1) \times (n+1)$
matrix. The matrix $C$ admits a unique
decomposition $C=U^{\dagger}DU$ where $U$ is a real, upper-triangular matrix
with 1's along its main diagonal and $D$ is a diagonal matrix with 
$D_{jj}=(2b)^{-2j}, \, j=0,1,\ldots,n$.
Hence, $F_n$ can be expressed as a sum of squares $F_n=q^{\dagger}Dq, \, q=Up$,
which can be explicitly written as
\begin{equation}
F_n(x,y,t) = \sum_{j=0}^n \left\vert\sum_{l=j}^n
\binom{l}{j}\frac{\partial_x^lp_n}{(2b)^l}\right\vert^2 =
\sum_{j=0}^n \left\vert\sum_{l=j}^n
\binom{l}{j}\left(\frac{i}{2b}\right)^lp_{n-l}\right\vert^2\,,     
\label{square}
\end{equation}
the last equality uses the fact 
$\partial_x p_n = i\partial_{\theta_1}p_n=ip_{n-1}$, which follows from 
\eqref{theta} and \eqref{propa}. It should be clear from the expression
of the generalized Schur polynomials in \eqref{pn} and the definitions
of the $\theta_1, \theta_2, \theta_3$ in \eqref{theta} that $F_n$ in \eqref{square}
is a positive definite polynomial in $x,y,t$ of degree $2n$, and it depends
on $2n+2$ real parameters, namely, $k := a+ib$ and 
$\gamma_j \in \mathbb{C}, \, j=1,2,\ldots,n$. Moreover, $F_n$ is also a
weighted homogeneous polynomial of degree $2n$ in 
$\theta_j, \bar{\theta}_j \, j=1,2,\ldots,n$ 
and $\Im(k)=b$ with weight$(\theta_j) =$ weight$(\bar{\theta}_j) =j$, and 
weight$(b)= -1$.
Consequently, the KPI multi-lump solution $u_n(x,y,t)$ given by \eqref{Fn}
is a non-singular, rational
function in the $xy$-plane for all values of $t$.
\paragraph{Remarks} 
\begin{itemize}
\item[(a)] One could exploit the linearity of the system
\eqref{phi} to choose 
$$\phi_n = \sum_{j=0}^n a_j(k)p_j(\theta_1,\cdots,\theta_j)\exp (i\theta(k))$$  
instead, to construct the $\tau$-function in \eqref{taun}. 
However, the above
sum can be reduced to a single generalized Schur polynomial
$p_n(\theta_1+h_1,\cdots,\theta_n+h_n)$ using \eqref{propc} for
suitable choices for the $h_j$'s which can then be absorbed in the arbitrary
constants $\gamma_j$ appearing in the $\theta_j$ variables for $j=1,\ldots,n$.
\item[(b)] The function $\phi_n = p_n\exp(i\theta)$ is also
a $\tau$-function for the KPI equation. But the corresponding rational
solution $u(x,y,t)$ given by \eqref{tau} is singular in the $xy$-plane for 
any given $t$. The singularities occur at the zeros of the generalized
Schur polynomial $p_n$.
\item[(c)] Since the KPI equation admits a constant solution 
$u(x,y,t)=c, \, c \in \mathbb{R}$, one may be interested in studying
the multi-lump solutions in a constant background. This is easily accomplished
by a few small modifications as follows: Equation \eqref{Fn}
should read as $u_n = c + 2(\ln \tau_n)_{xx}$ where $\tau_n$ is given 
by \eqref{taun} except that the $\phi_n$ satisfies the linear system:\,
$i\phi_y = \phi_{xx} + c\phi, \,\,  \phi_t + \phi_{xxx} + \sf34 c\phi_x=0$.
The solution to this new linear system is given by 
$\phi_n = \sf{1}{n!}\,\partial_k^n\exp(i\theta') = p'_n\exp(i\theta')$ where
$\theta' = kx+(k^2-c)y+(k^3-\sf34 kc)t+\gamma(k)$.
The new generalized Schur polynomials $p'_n$ are defined in the same way
as before, the only change being $x \to x-\sf34 ct$ in the definition
of $\theta_1$ in \eqref{theta}.
\item[(d)] Several authors~\cite{DM11,C18} have considered 
a more general choice for $\phi_n$
such as $\phi_n = D_k^n\exp(i\theta)$ with $\theta(k)$ same as
before but $D_k := f(k)\partial_k$ where $f(k)$ is analytic. In such cases,
it is possible to consider $D_k = \partial_z$ in terms of
a (local) uniformizing variable $z(k)$ with $z'(k)=1/f(k)$
and an appropriate branch of its inverse $k(z)$ to express $\theta(k(z))$.
Thus this approach is consistent with this section's formalism although
the corresponding rational KPI solutions are different due to differing
choices for the generalized Schur polynomials $p_n(z)$.
\end{itemize}
\section{Examples of multi-lump solutions}
We now give some examples of the KPI multi-lump solutions and illustrate
the interaction properties of such solutions. It will be convenient
to first introduce a set of coordinates $r,s$ defined in terms of
co-moving coordinates $x',y'$ as follows 
\begin{equation}
r=x'+2ay', \quad s = 2by', \quad \text{where} \quad
x'=x-3(a^2+b^2)t, \quad y'=y+3at \,, 
\label{rs}
\end{equation}
and recall that $a=\Re(k), b=\Im(k)$. As will be shown below that $r,s$
are the natural coordinates to use instead of the $x,y$. The $\theta_j$
variables defined in \eqref{theta} are then given by
\begin{equation}
\theta_1 = i(r+is+\gamma_1),  \quad \theta_2 = \frac{is}{2b} -3bt +i\gamma_2,
\quad \theta_3=i(t+\gamma_3)
\label{thetars}
\end{equation}
\begin{figure}[h!]
\begin{center}
\raisebox{0.1in}{\includegraphics[scale=0.65]{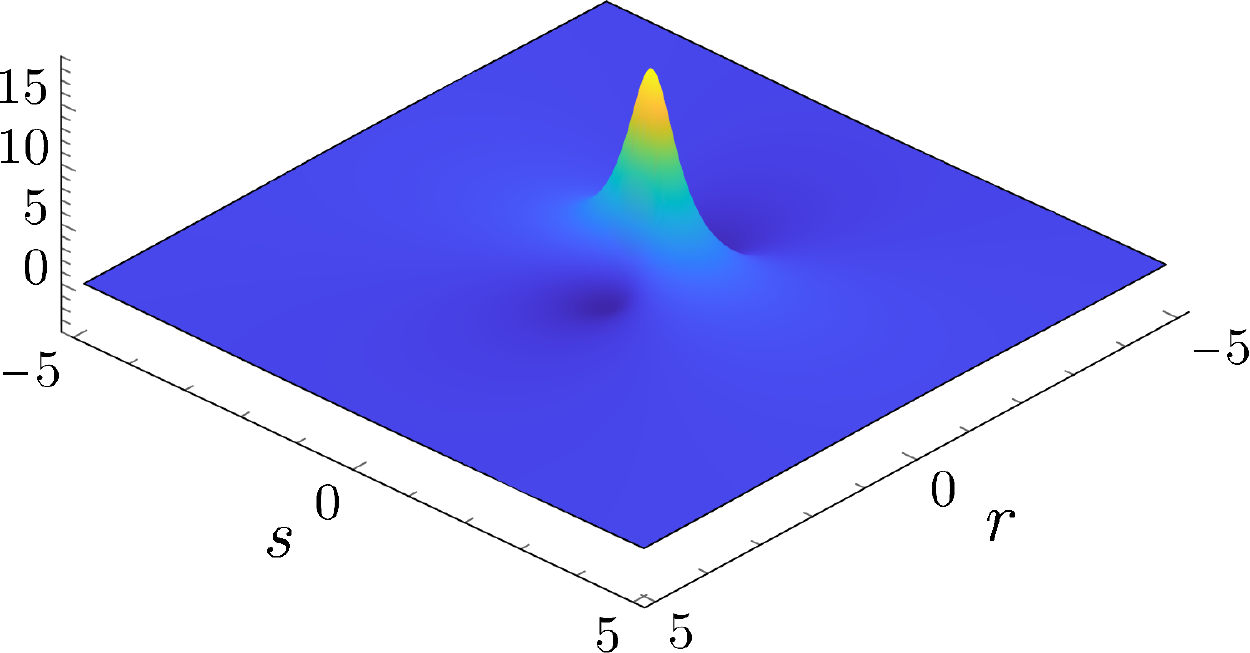}} \qquad
\includegraphics[scale=0.6]{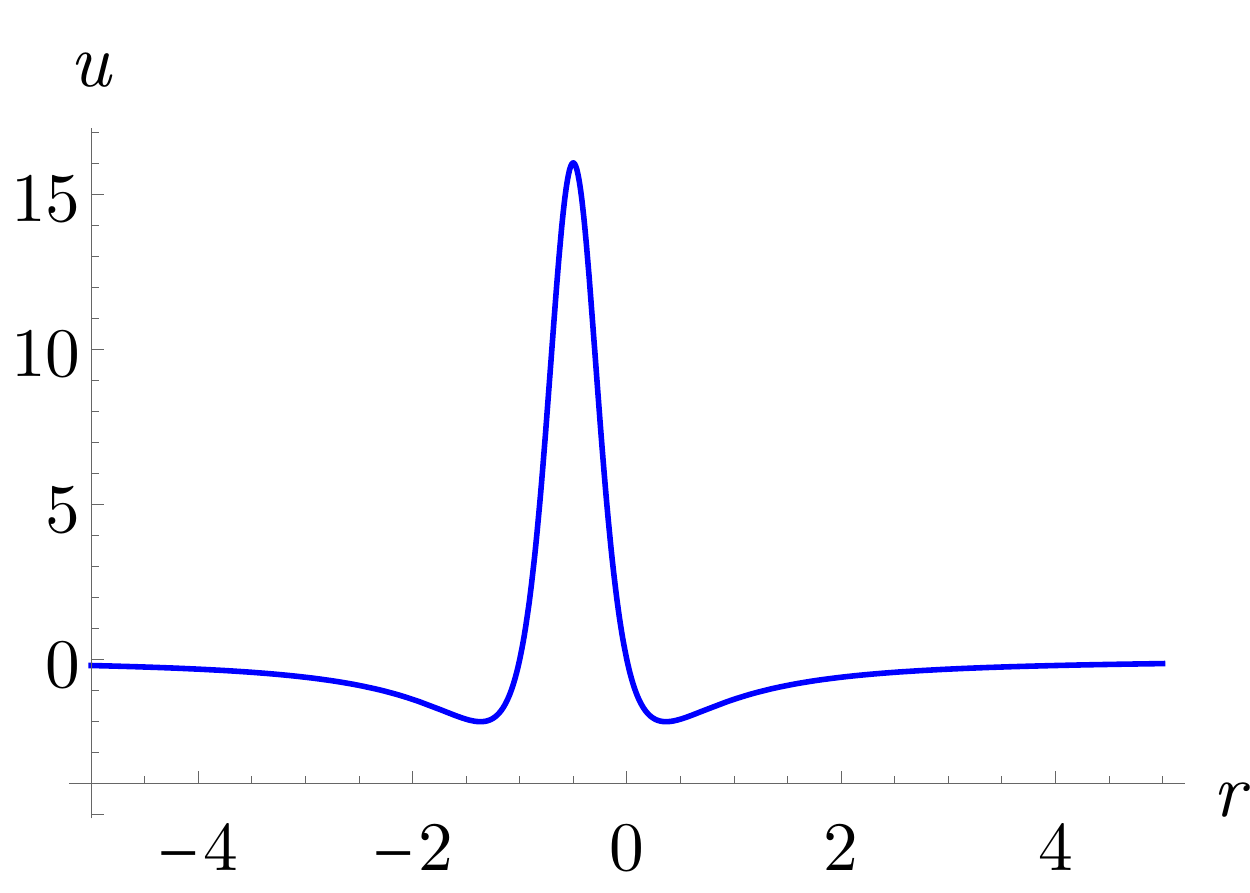}
\end{center}
\vspace{-0.2in}
\caption{$1$-lump solution of the KPI equation:\, $a=\gamma_1=0, b=1$. 
Right: Vertical crossection showing the maximum and two minima.}
\label{1lump}
\end{figure}
\subsection{1-lump solution} For $n=1$, the polynomial $F_1$ is obtained
from \eqref{square} after using $p_0=1, p_1=\theta_1$ as
\[ F_1= |\theta_1+\sf{i}{2b}|^2 +\sf{1}{4b^2} = 
(r+\sf{1}{2b}-r_0)^2+(s-s_0)^2 + \sf{1}{4b^2}, 
\qquad r_0=\Re(\gamma_1), \,\, s_0=\Im(\gamma_1) \,,\]
where \eqref{thetars} is used to obtain the second expression of $F_1$.
Then from \eqref{Fn}, the 1-lump solution is given by
\begin{equation}
u_1(x,y,t) = 2(\ln F_1)_{xx} =  
4\,\frac{-(r+\sf{1}{2b}-r_0)^2+(s-s_0)^2+\sf{1}{4b^2}}
{\big[(r+\sf{1}{2b}-r_0)^2+(s-s_0)^2+\sf{1}{4b^2}\big]^2}\,.    
\label{u1}
\end{equation}
Notice that the solution is {\it stationary} in the $rs$-plane since
the only time dependence enters via the comoving coordinates $x',y'$.
Hence, the 1-lump is a rational traveling wave with a single peak (maximum)
at $(r_0-\sf{1}{2b}, s_0)$ of height $16b^2$ and two local minima symmetrically
located from the peak at $(r_0-\sf{1}{2b}\pm \sf{\sqrt{3}}{2b})$ and
depth $-2b^2$ determined by $b$ and
the complex parameter $\gamma_1$ as illustrated in Figure~\ref{1lump}.
The wave moves in the $xy$-plane with
a uniform velocity $(3(a^2+b^2), -3a)$ at an angle $\tan^{-1}(-a/(a^2+b^2))$
with the positive $x$-axis. Since $u_1$ is the $x$-derivative of
the rational function $F_{1x}/F_1$ which decays as $|x| \to \infty$ for all $y,t$, one has
that $\int_{-\infty}^\infty u_1\,dx=0$. However, in the $xy$-plane $u_1$ is not
integrable in the Fubini-Tonelli sense. Hence, $u_1$ is not a $L^1(\mathbb{R}^2)$
function although $u_1 \in L^2(\mathbb{R}^2)$ and 
$\int\!\!\int_{\mathbb{R}^2}u_1^2 = 16\pi b$; the latter is a conserved 
quantity for the KPI equation \eqref{kp}.  

\subsection{\bf 2-lump solution}    
In this case the polynomial $F_2$ from \eqref{square} is given by 
\begin{equation}
F_2 = |\sf12((r+\sf{1}{2b})^2-s^2)+3bt+\sf{1}{8b^2}+ irs|^2 + 
\sf{1}{4b^2}|(r+\sf{1}{b})+is|^2+\sf{1}{16b^4} \,,
\label{F2}
\end{equation}
where we have set the constants $\gamma_1=\gamma_2=0$ for simplicity.
Notice that unlike $F_1$, the polynomial $F_2(r,s,t)$ does depend explicitly on $t$
so that the solution $u_2$ obtained from \eqref{Fn} 
is non-stationary in the comoving $rs$-plane. The explicit expression
for $u_2(x,y,t)$ is complicated, so it is not included here.
Figure~\ref{2lump}  illustrates that
$u_2$ consists of 2 localized lumps along the $r$-axis that are well 
separated as $t \ll 0$; these lumps get attracted to each other
and overlap when $t$ is finite, then separate
when $t \gg 0$ but along the $s$-axis. Furthermore, the height of each peak
also evolve with time and approaches the constant height of the 1-lump
solution as $|t| \to \infty$. The interaction process is an example of 
anomalous (inelastic) scattering rather than the usual solitonic interaction
of the simple $n$-lump solution of KPI.
\begin{figure}[t!]
\begin{center}
\includegraphics[scale=0.42]{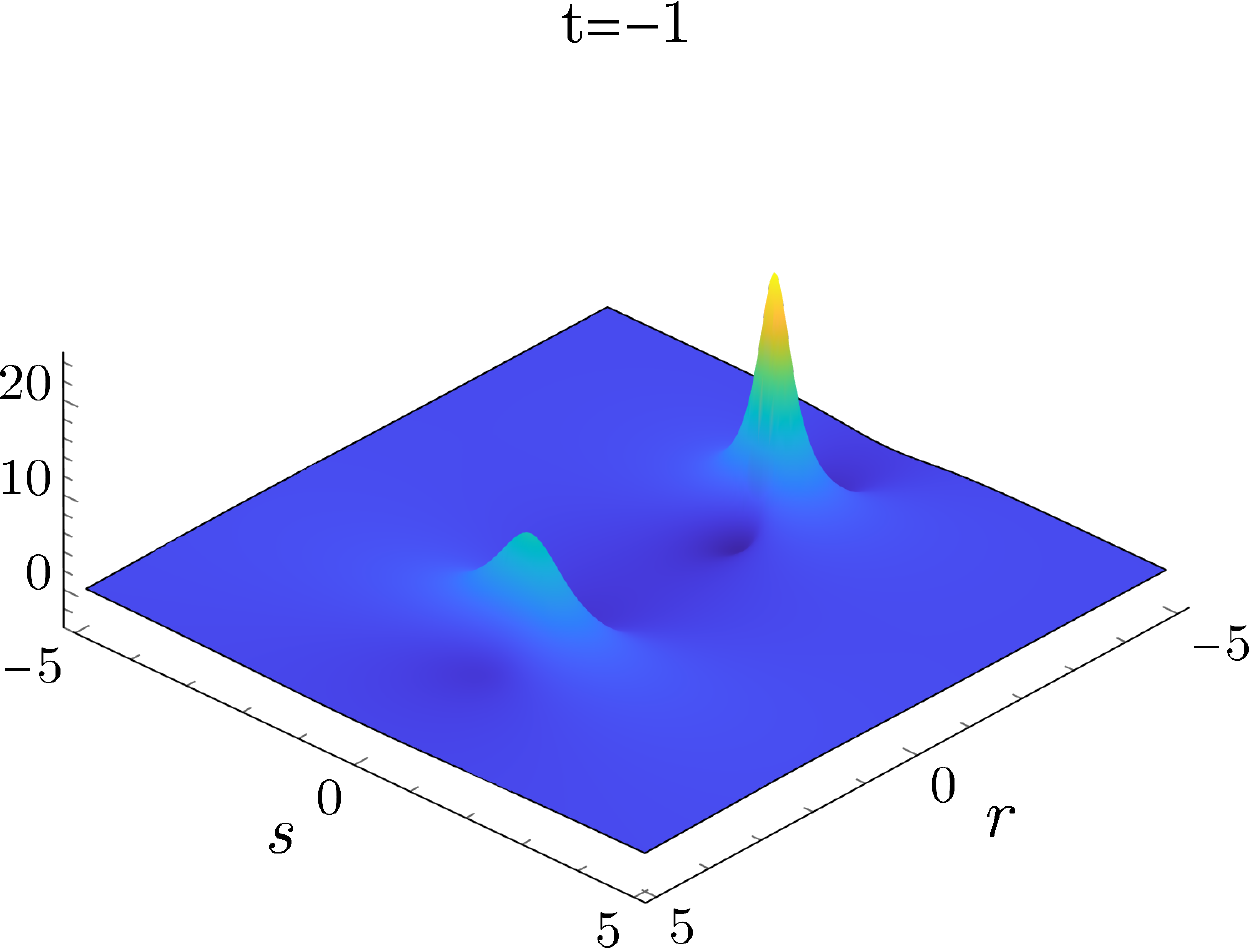} \quad 
\includegraphics[scale=0.42]{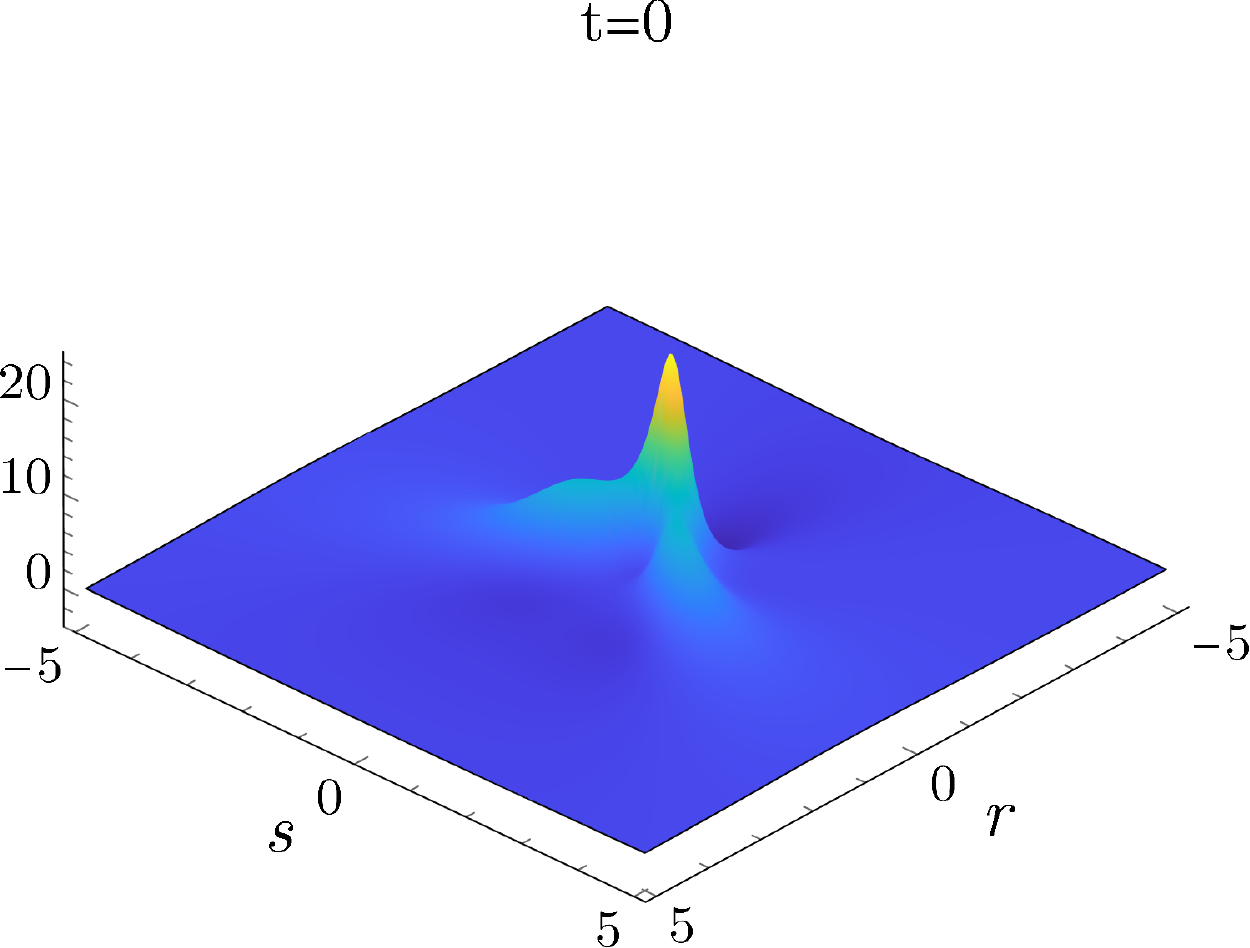} \quad 
\includegraphics[scale=0.42]{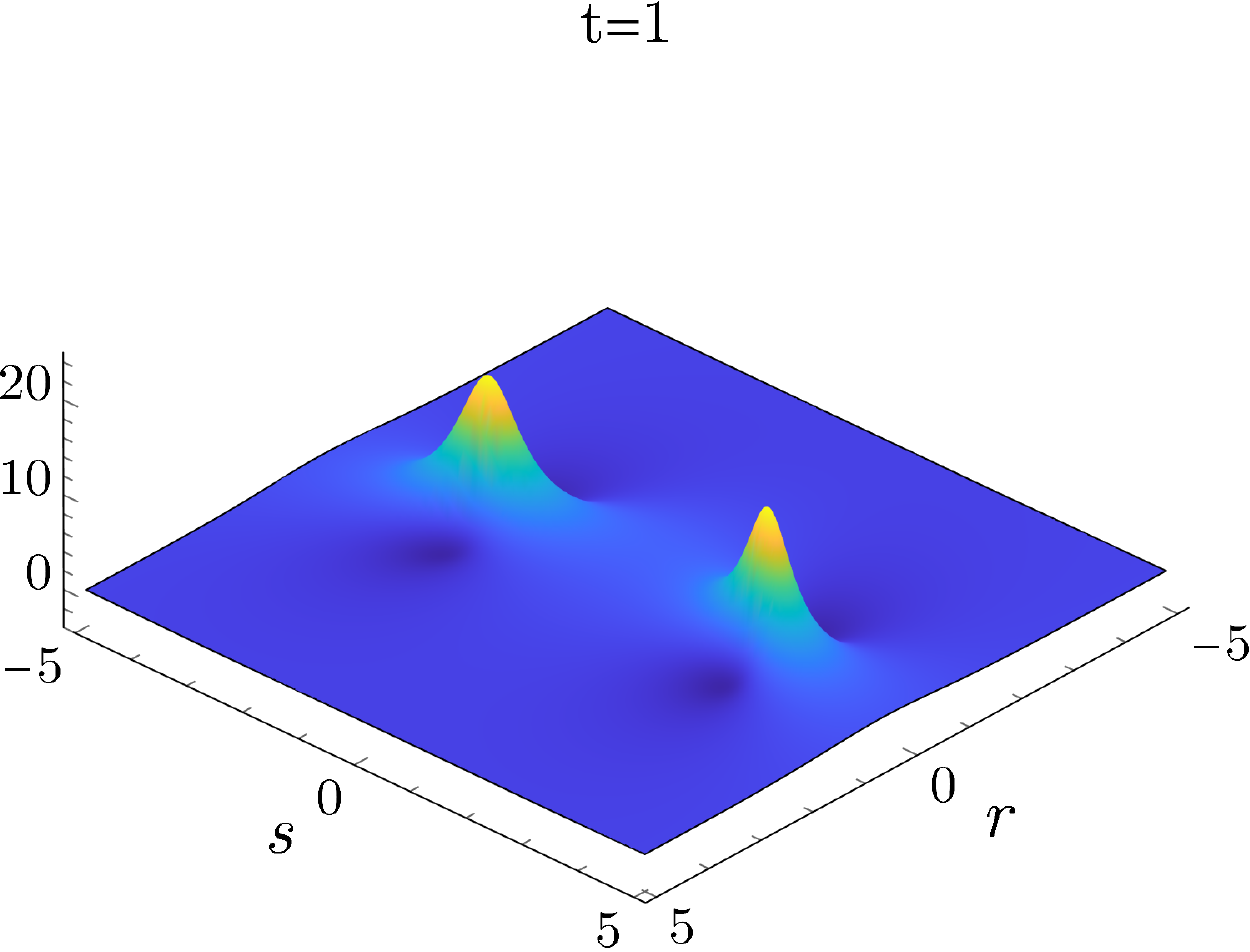}
\end{center}
\vspace{-0.2in}
\caption{$2$-lump solution of the KPI equation:\, $a=\gamma_1=\gamma_2=0, b=1$.}
\label{2lump}
\end{figure}
In order to analyze the $2$-lump interaction, the peak locations and heights
need to be determined from the local maxima of the function $u_2(x,y,t)$
but the necessary calculation is too complicated to yield exact analytical formulas.
Instead, we make certain approximations to estimate the leading order expressions
for peak locations and heights. First, we note that the expression
\[u = 2 \ln(F)_{xx} = 2\frac{FF_{xx}-F_x^2}{F^2} \]
for the solution suggests that the maximum for $u_2$ occurs approximately
near the minimum of the positive polynomial $F_2$. Secondly, the local
minima of $F_2$ are approximately located near the zeros of the leading
order polynomial $F_2^{(0)}$ which is the first of the three
$| \cdot |^2$  terms in \eqref{F2}
corresponding to the $j=0$-term in \eqref{square}. This assumption, which
will be justifed below (see also Appendix), makes the analytical computations tractable
and corroborates the interaction process depicted in Figure~\ref{2lump} above.
Setting $F_2^{(0)}=0$ and after solving for $r$ and $s$,
one obtains the approximate peak locations
\begin{equation}
(r_p, s_p) = (-\sf{1}{2b}\pm\sqrt{6b|t|-\sf{1}{4b^2}},\, 0), 
\quad t \ll 0, \qquad
(r_p, s_p) =  (0,\, \pm\sqrt{6bt+\sf{1}{2b^2}}), \quad t \gg 0 \,.
\label{F2zero}
\end{equation}
\begin{figure}[h!]
\begin{center}
\includegraphics[scale=0.42]{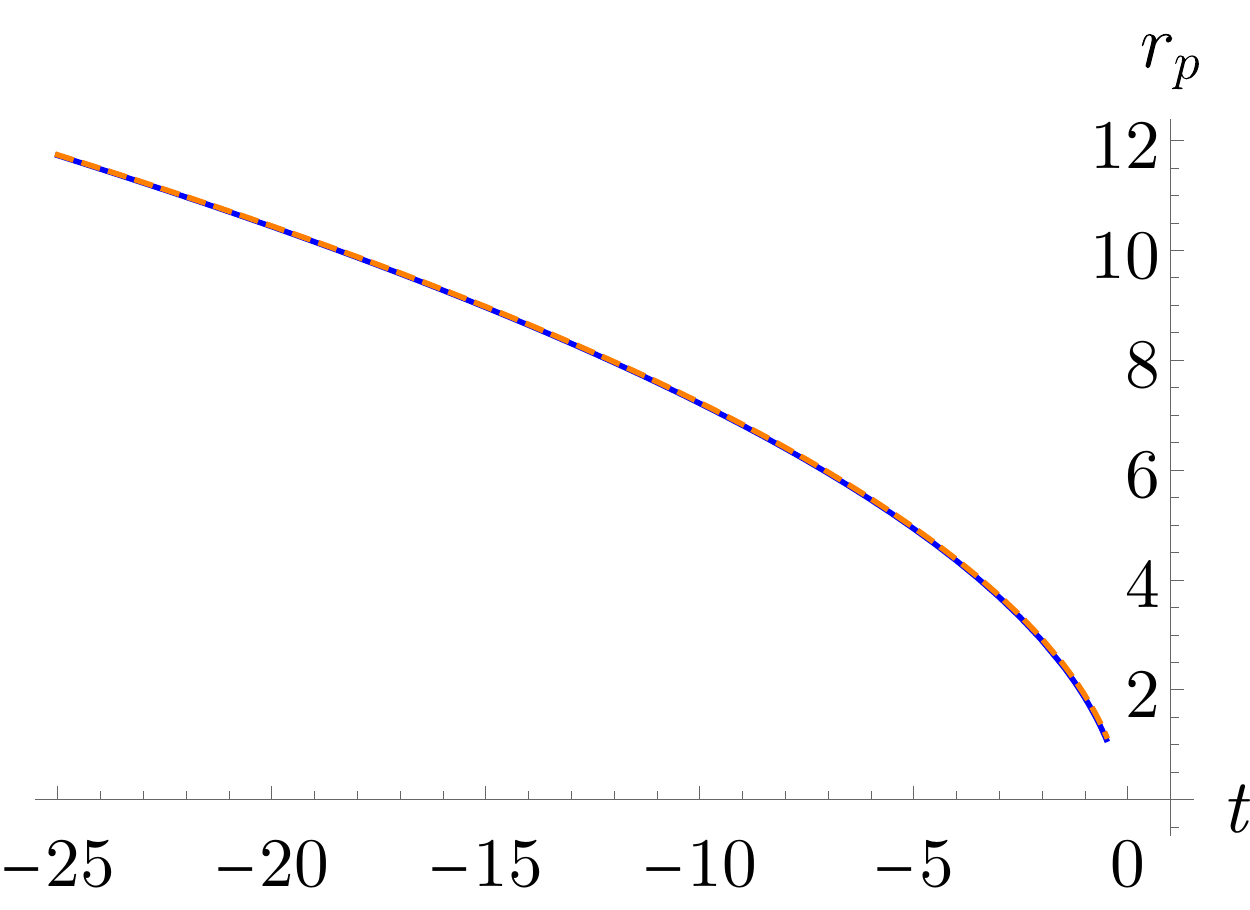} \quad 
\raisebox{0.3in}{\includegraphics[scale=0.42]{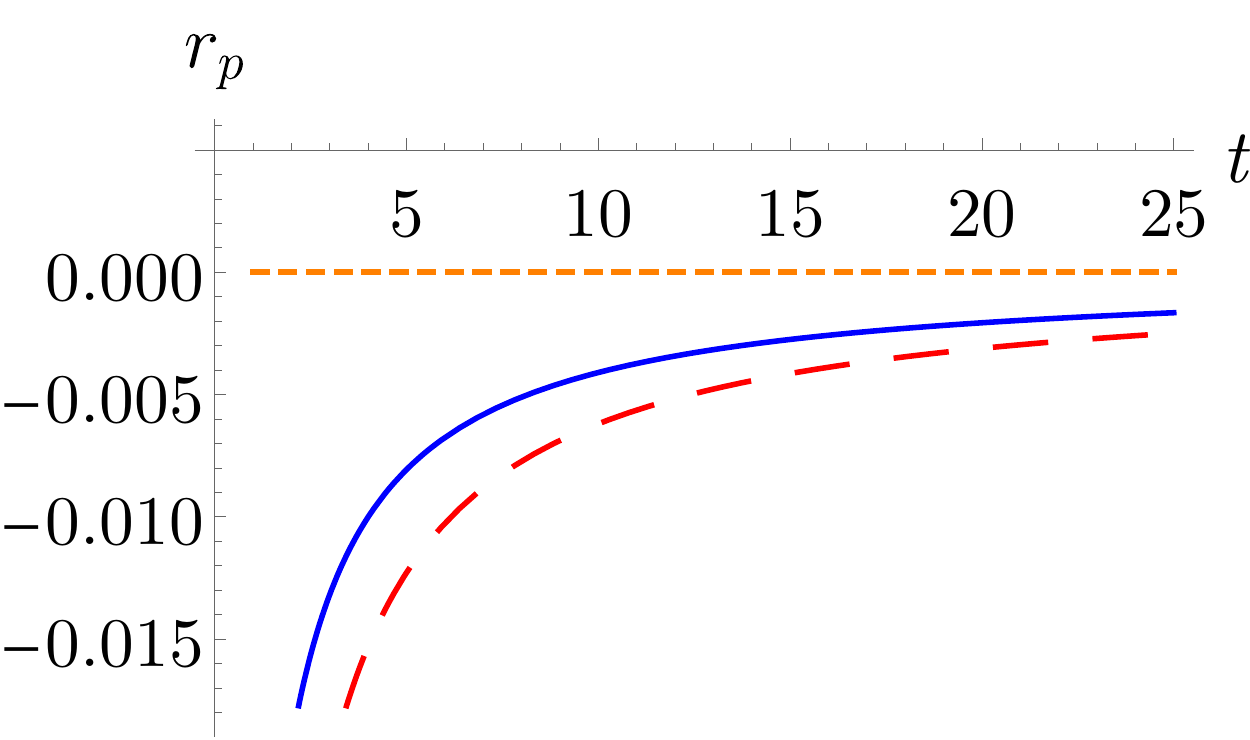}} \quad
\includegraphics[scale=0.42]{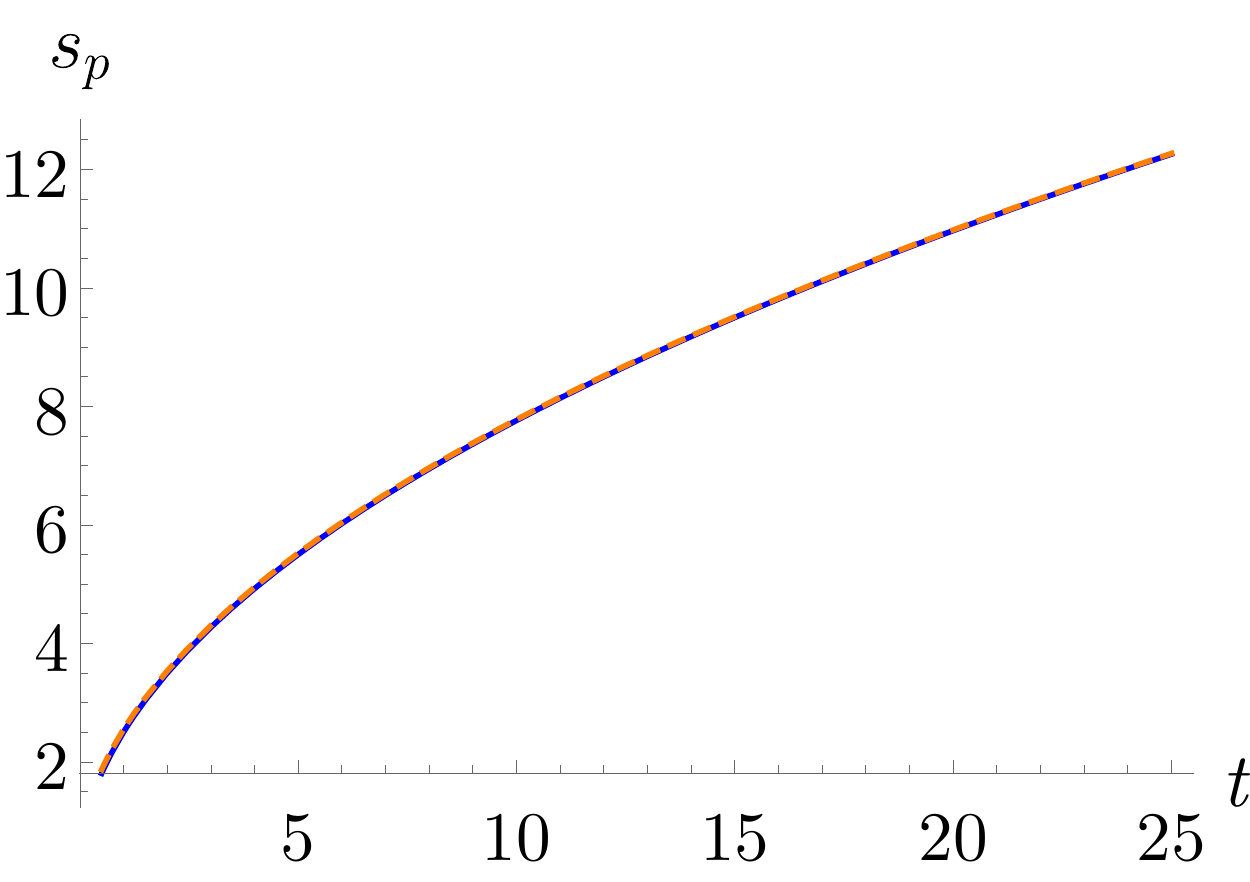}
\end{center}
\vspace{-0.2in}
\caption{$2$-lump peak locations. Exact (blue curve), approximate from
\eqref{F2zero} (orange 
dashed curve). Left panel:\, $r_p(t)$ for right peak ($s_p=0$). 
Reflecting across the line $r_p=-\sf{1}{2b}$ gives the left peak location.
Middle panel:\, $r_p(t)$ for top \& bottom peak, red-dashed line corresponds
to minimum of $F_2$. Right panel:\, $s_p(t)$ for top peak, reflecting 
across $t$-axis gives $s_p(t)$ of the bottom peak.
KP parameters: $a=\gamma_1=\gamma_2=0, b=1$.}
\label{2lumppeaks}
\end{figure}

Equation \eqref{F2zero} shows that the peaks are well-separated for $|t| \gg 0$
and verifies the anomalous interaction process
described above. It is also evident from \eqref{F2zero}
that the relative rates of attraction
(along $r$-axis) and repulsion (along $s$-axis) is $\sqrt{6b}|t|^{-1/2}$
as $|t| \to \infty$. The exact and approximate peak locations 
from \eqref{F2zero} are almost indistinguishable as shown in 
Figure~\ref{2lumppeaks} which illustrates the accuracy of \eqref{F2zero}
for large $t$.  The exact and approximate values differ
by $O(|t|^{-1/2})$ as will be shown below (and in the Appendix
for the general $n$-lump solutions). This fact is manifested
in the middle panel of Figure~\ref{2lumppeaks} which shows that the 
approximate value of $r_p(t)$ for $t \gg 0$, obtained by minimizing $F_2$ 
gives a better approximation $r_p(t)=-\sf{1}{16b^4}\sf{1}{t}+O(t^{-2})$
than $r_p(t)=0$ given by \eqref{F2zero}.
The difference between the two approximations is $O(t^{-1})$.
From \eqref{rs}, the peak locations in the $xy$-plane
are given by $(x_p(t), y_p(t)) = 
(r_p(t) - 3(a^2+b^2)|t|,\, 3a|t|)$ when $t \ll 0$, 
and the corresponding trajectory is a parabola which degenerates to the $x$-axis
when the parameter $a=0$. For $t \gg 0$, 
$(x_p(t), y_p(t)) = (3(a^2+b^2)t-\sf{a}{b}s_p(t),\, -3at+\sf{s_p(t)}{2b})$,
and the trajectory is again a parabola in the $xy$-plane.
If $a=b$, this parabola degenerates
to a straight line $x+2ay=0, \, a \neq 0$. A plot of a $2$-lump trajectory
in the $xy$-plane for $a=0$ is given in Figure~\ref{2lumptraj}.

The approximate peak heights can be calculated from the formula
$u_2(r_p,s_p)$ but since we assume that $F$ attains its minimum
value near $(r_p,s_p)$, we can take $u_2 \approx 2F_{2xx}/F_2$. Using 
\eqref{F2} and \eqref{F2zero}, we obtain the following expressions
for the approximate peak heights
\begin{equation}
u^{\pm}_{2p} = \frac{16b^2}{1+\epsilon \pm 2\sqrt{\epsilon-\epsilon^2}},
\quad t \ll 0, \quad \qquad
u_{2p} = 16b^2 - \frac{48b^2 \epsilon}{1+7\epsilon}, \quad t \gg 0\,,
\label{u2p}
\end{equation}
where $\epsilon := (24b^3|t|)^{-1}$ and $u_{2p}^{-}, u_{2p}^+$
denote the left and right peak heights, respectively. 
The peak heights approaches the $1$-lump
peak height $16b^2$ asymptotically as $|t| \to \infty$. For $t \ll 0$, as the 
lumps approach each other along the $r$-axis, the left peak 
grows while the right peak decreases with time. But when
$t \gg 0$, both peaks $u_{2p}$ have the same value which increases with time
and approaches the value $16b^2$. These features are illustrated by 
Figure~\ref{2lumppeakheights} below.  
\begin{figure}[t!]
\begin{center}
\includegraphics[scale=0.7]{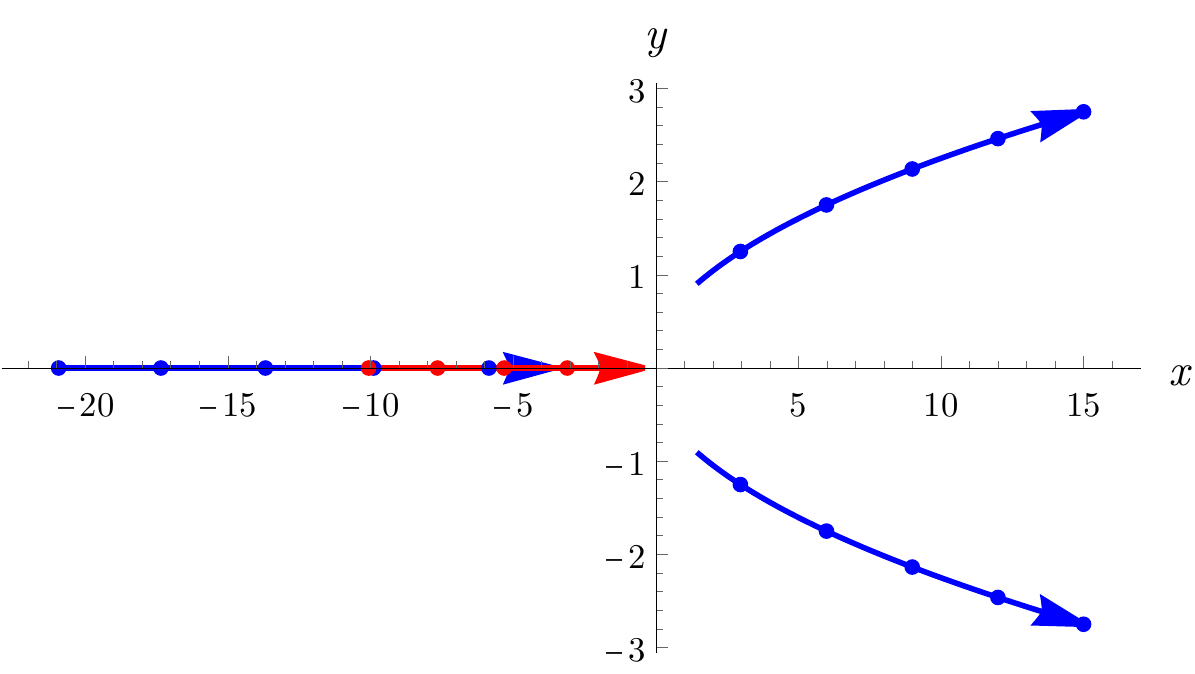} 
\end{center}
\vspace{-0.2in}
\caption{$2$-lump peak trajectories in the $xy$-plane with $a=0$. For
$t \ll 0$, the trajectories degenerate along the negative $x$-axis with
the blue dots indicating the faster peak. For $t \gg 0$, the parabolic
trajectories are shown on the right-half $xy$-plane.}
\label{2lumptraj}
\end{figure}
\begin{figure}[h!]
\begin{center}
\includegraphics[scale=0.42]{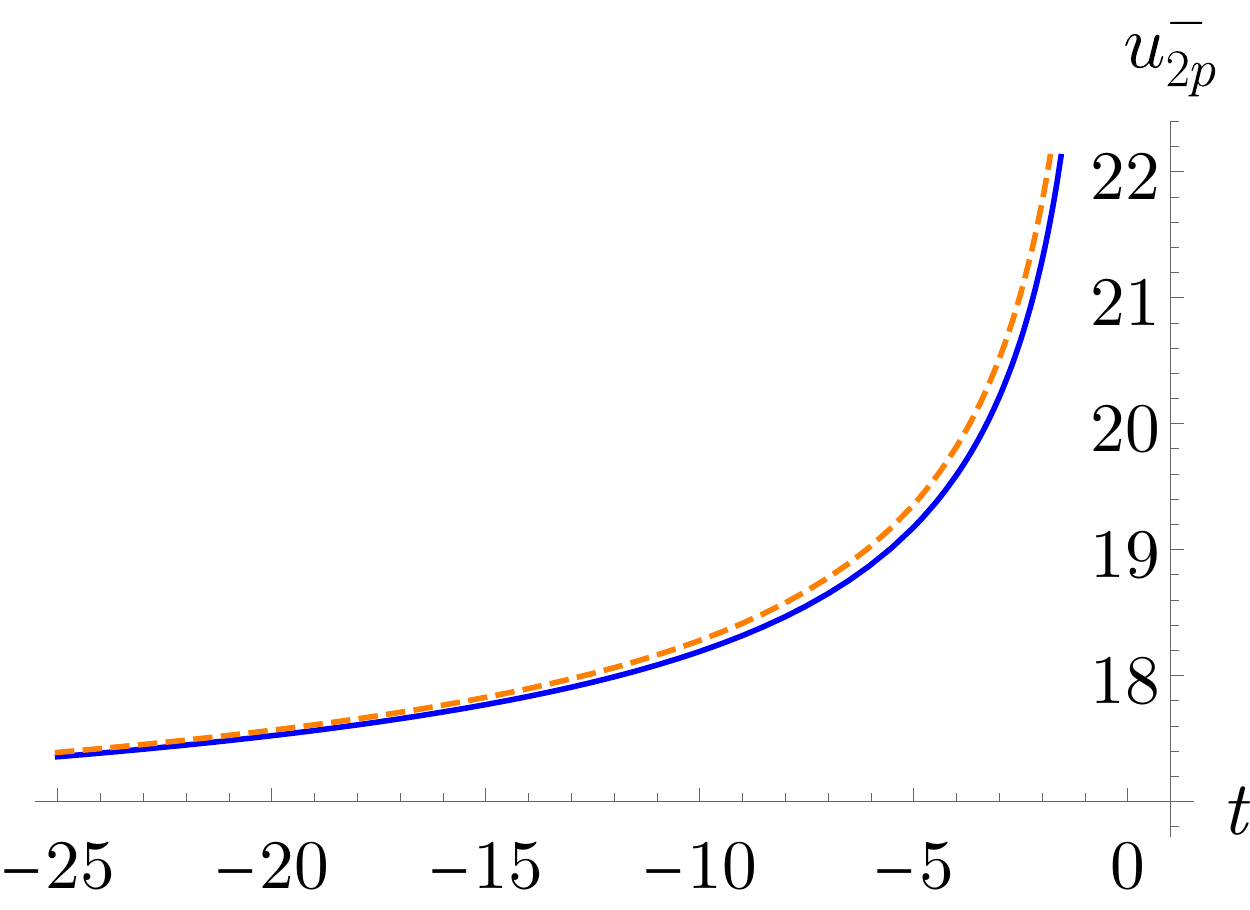} \quad
\includegraphics[scale=0.42]{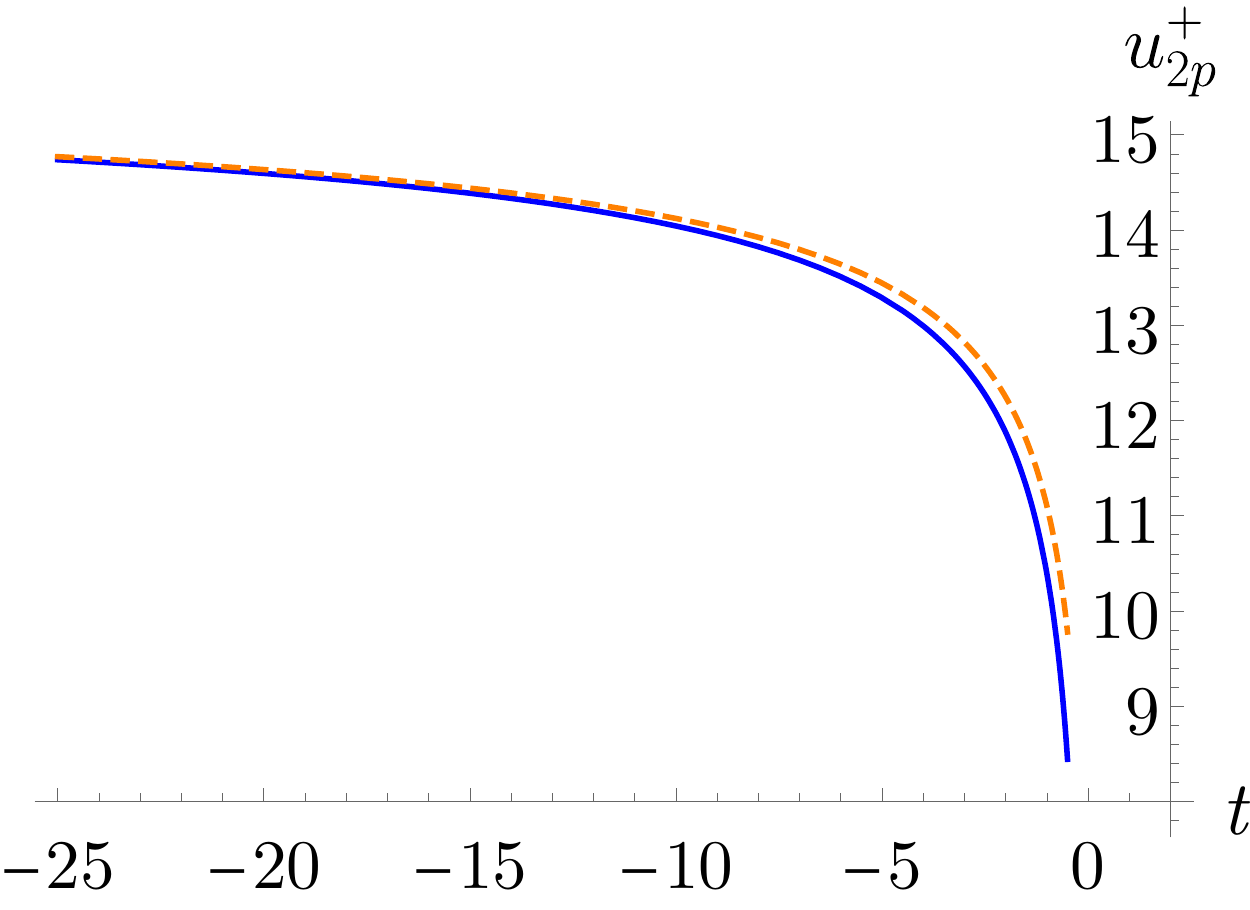} \quad
\includegraphics[scale=0.42]{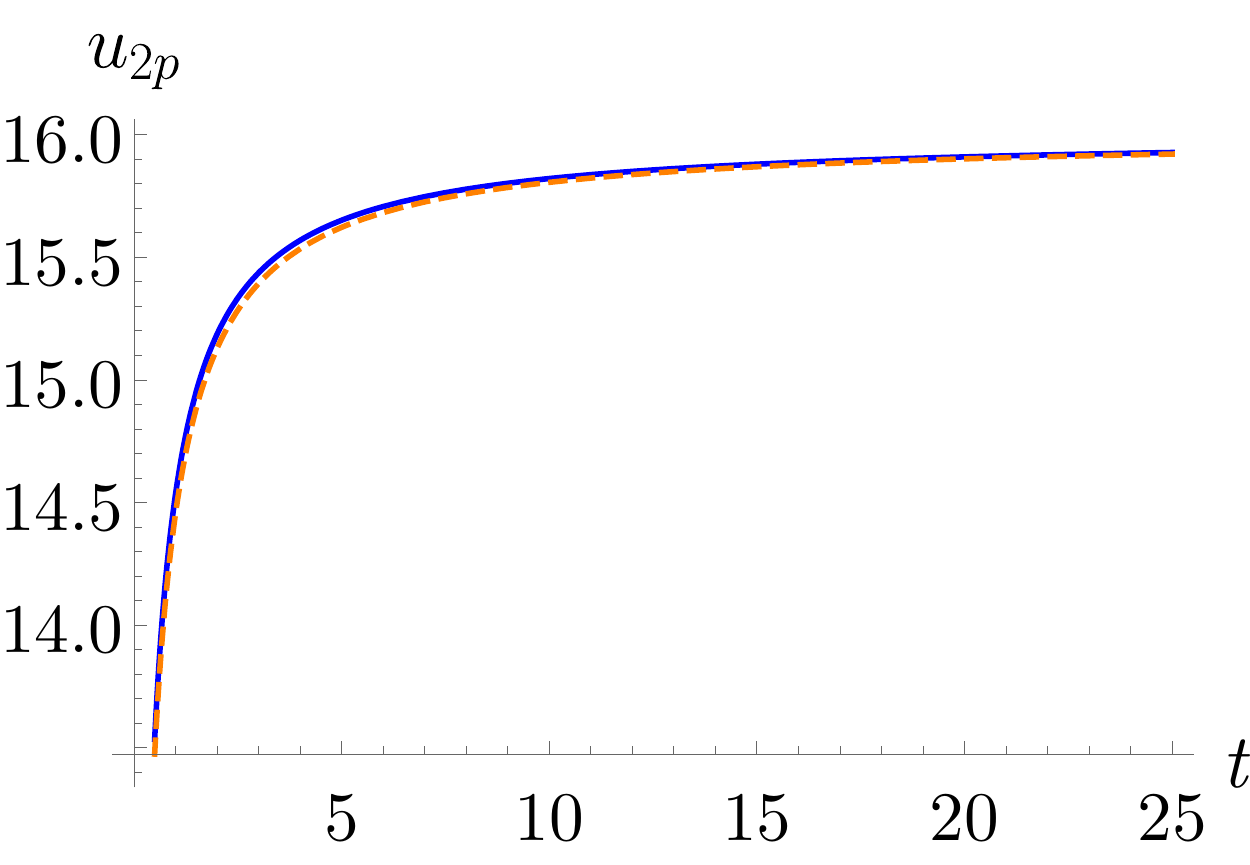}
\end{center}
\vspace{-0.2in}
\caption{$2$-lump peak heights: $u_{2p}^{\pm}$ vs $t$ (left two panels),
$u_{2p}$ vs $t$ (right panel). Exact heights -- solid blue line,
approximate heights (see \eqref{u2p}) -- orange dashed line.}
\label{2lumppeakheights}
\end{figure}

To compare with \eqref{F2zero}, 
we also list below the approximate peak 
locations $(r_p,s_p)$ obtained by minimizing $F_2$
\begin{align*}
(r_p,s_p)&= \big(\pm \sqrt{6b|t|}-\sf{1}{2b} \mp 
\sf{\sqrt{6}}{16b^{5/2}}\sf{1}{|t|^{1/2}} + O(\sf{1}{|t|}), \quad 0\big), 
\qquad t \ll 0, \\
(r_p,s_p)&= \big(-\sf{1}{16b^4}\sf{1}{t}+O(\sf{1}{t^2}),\quad \pm \sqrt{6bt} 
+ O(\sf{1}{t^{3/2}})\big), \qquad t \gg 0\,, 
\end{align*}
as well as the exact values obtained by maximizing $u_2$
\begin{align*}
(r_p,s_p)& = \big(\pm\sqrt{6b|t|} -\sf{1}{2b}
\mp\sf{\sqrt{6}}{24b^{5/2}}\sf{1}{|t|^{1/2}}+O(\sf{1}{|t|^{3/2}}), \quad 0\big), 
\qquad t\ll 0, \\
(r_p,s_p)& = \big(-\sf{1}{24b^4}\sf{1}{t} +O(\sf{1}{t^2}), \quad
\pm \sqrt{6bt} \pm \sf{\sqrt{6}}{48b^{5/2}}\sf{1}{t^{1/2}}
+O(\sf{1}{t^{3/2}})\big), \qquad t \gg 0  \,. 
\end{align*}
It is clear that \eqref{F2zero} gives the leading order peak locations
and differ from the exact values by $O(|t|^{-1/2})$. Hence, \eqref{F2zero}
can be employed to examine the asymptotic behavior of the 2-lump solution
for large $|t|$. Substituting $r=r_p+h, s=s_p+k$ into the
expression for $F_2$ in \eqref{F2}, using the first expression
from \eqref{F2zero}, and retaining
the leading order terms in $|t|$ for $t \ll 0$, yields
\[F_2 = 6b|t|\big[h^2+k^2+\sf{1}{4b^2} + O(|t|^{-1/2})\big] \,.\]
The leading order expression inside the square brackets
of the above equation is precisely the $1$-lump polynomial $F_1(h,k)$ given  
above \eqref{u1} with $h=r-r_p, k=s-s_p$. Then,
asymptotically as $t \to -\infty$, the $2$-lump solution $u_2$ can
be viewed as a superposition of two $1$-lump solutions whose peaks
are located at $(r,s) = (\pm\sqrt{6b|t|}, \, 0)$, each with height $16b^2$.      
A similar calculation shows also that $u_2$ is a superposition of two 
$1$-lumps with peaks at $(r,s)=(0, \, \pm\sqrt{6bt})$ and height $16b^2$, 
asymptotically as $t \to \infty$ so that the $L^2$-norm
$\int\!\!\int_{\mathbb{R}^2}u_2^2 = 2(16\pi b)$ for all $t$.
\paragraph{Remarks}  
\begin{itemize}
\item[(a)] If the parameters $(\gamma_1,\gamma_2) \neq (0,0)$
then each of the peak locations $(r_p,s_p)$ shifts by a constant $(r_0,s_0)$
and the trajectories of the peaks are hyperbolas in
the $rs$-plane given by $(r_p-r_0)(s_p-s_0) =$ constant depending on 
$\gamma_1, \gamma_2$. Otherwise, the asymptotic behavior of the peaks
remain essentially the same as above.
\item[(b)] The $2$-lump dynamics was discussed in earlier
studies~\cite{GPS93,VA99} as well.
But a more detailed discussion of the asymptotic
dynamics of the $2$-lump solution including the time evolution 
of the peak heights, are presented in this paper.
\end{itemize}
\begin{figure}[h!]
\begin{center}
\includegraphics[scale=0.44]{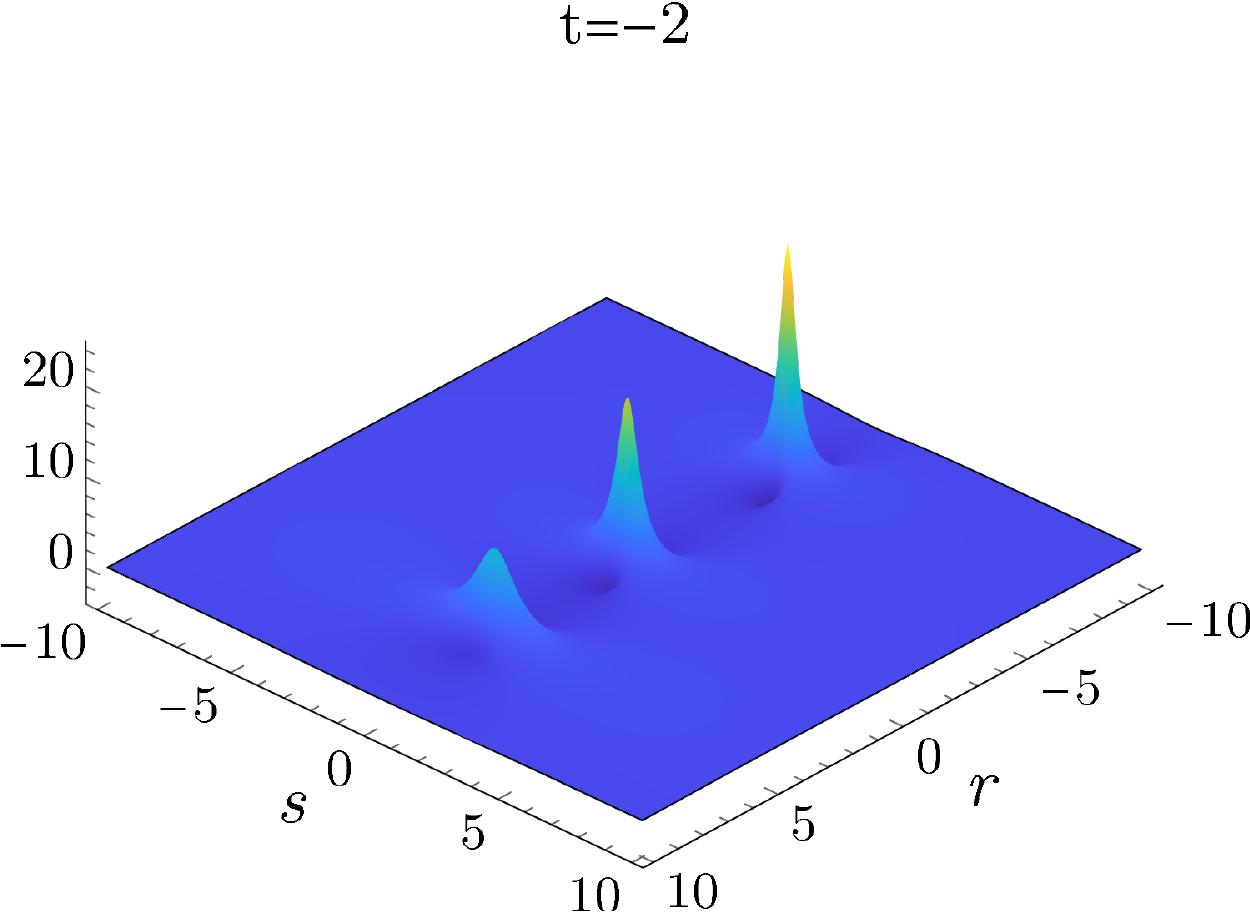} 
\includegraphics[scale=0.44]{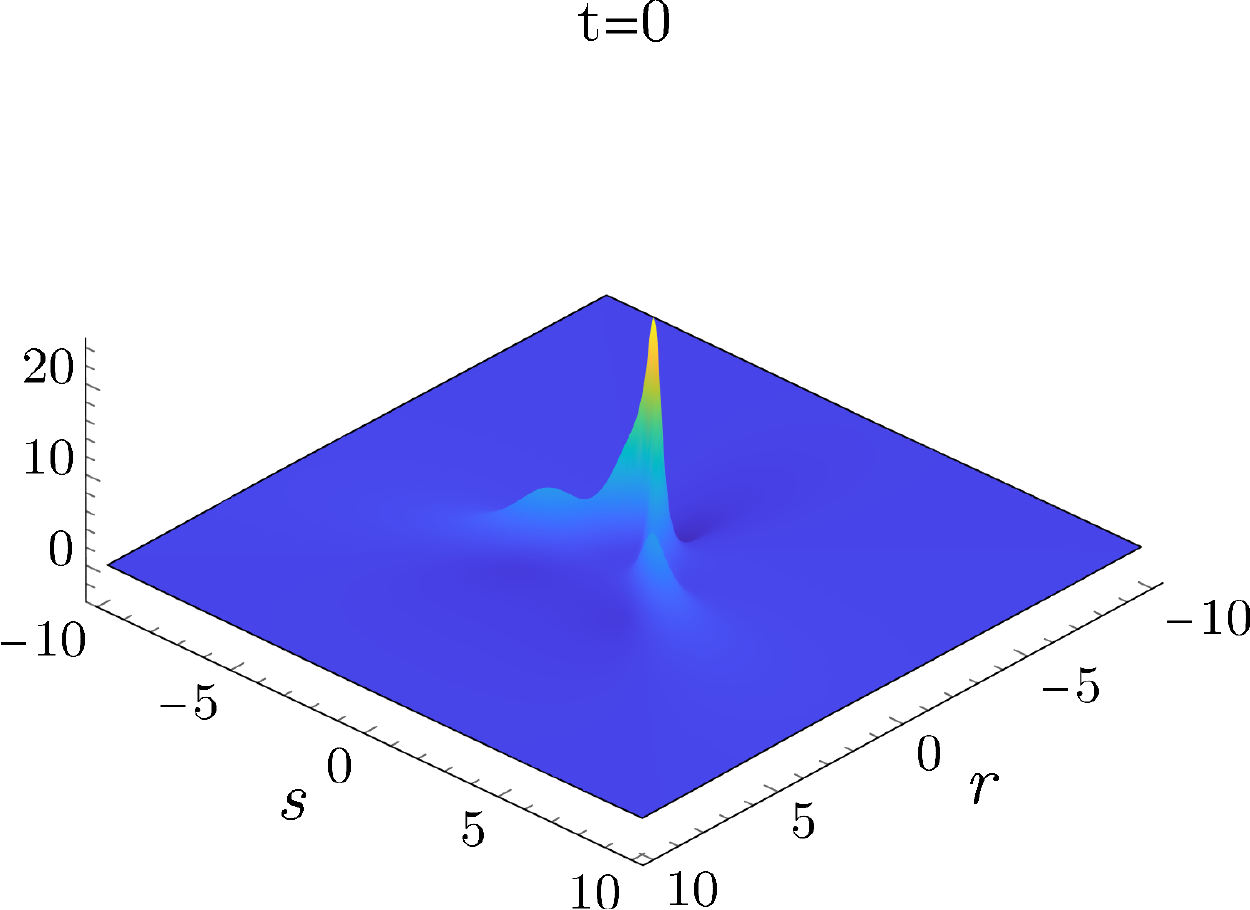} 
\includegraphics[scale=0.44]{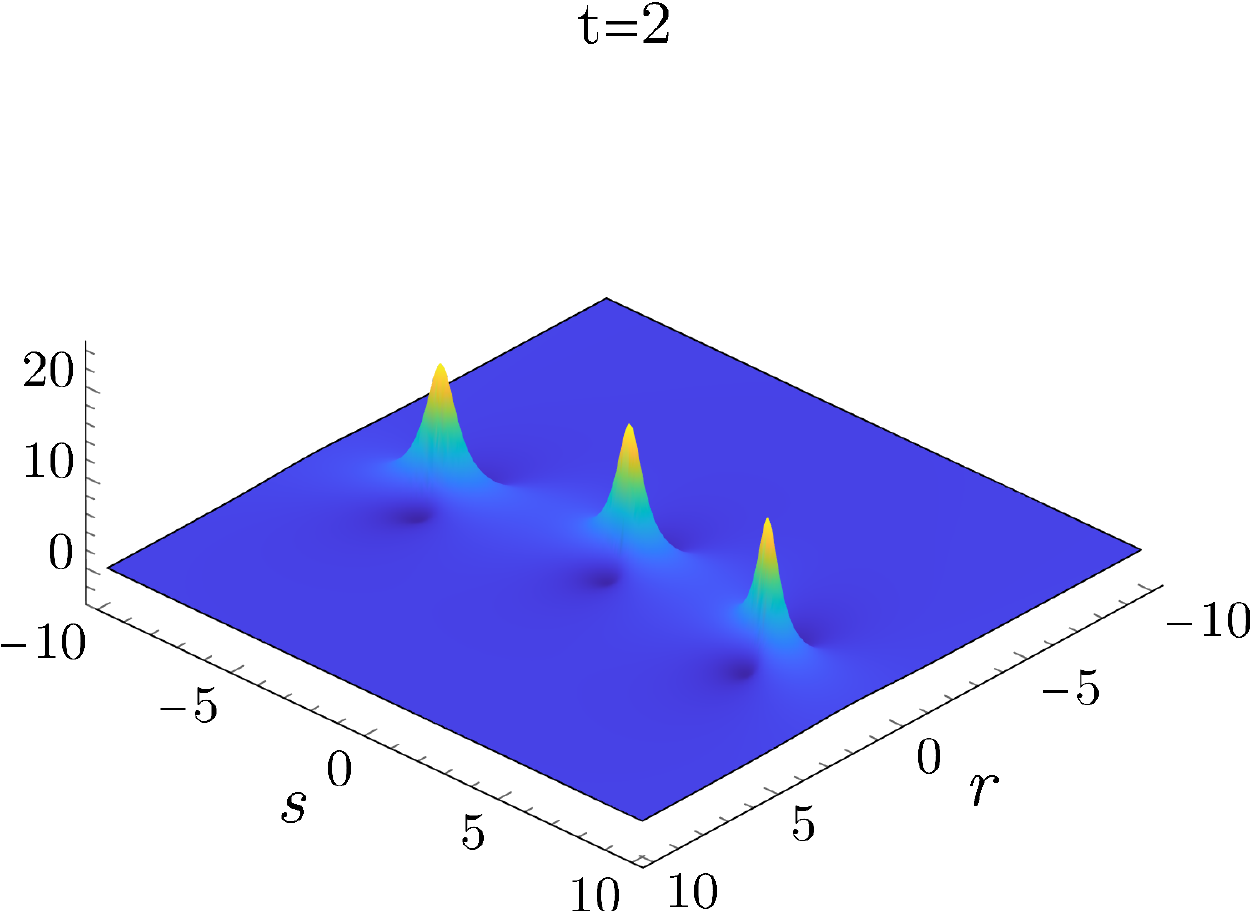}
\end{center}
\vspace{-0.2in}
\caption{$3$-lump solution of the KPI equation:\, 
$a=\gamma_1=\gamma_2=\gamma_3=0, b=1$.}
\label{3lump}
\end{figure}
\subsection{$3$-lump solution}
The $3$-lump solutions and its asymptotic dynamics can be analyzed in
the same manner as in the $2$-lump case. Therefore, we omit most of the details
and illustrate the dynamics via Figure~\ref{3lump} which shows
the anomalous scattering of a $3$-lump solution as the peaks
approach along the $r$-axis and scatter along the perpendicular
$s$-axis. The main feature of this solution (and any odd $n$-lump)
is that the central peak remains on the $r$-axis throughout
the interaction and moves at a much slower rate than 
the peaks on its either side. 

Equation \eqref{square} with
$n=3$ gives the polynomial $F_3$ whose leading ($j=0$) term is 
\begin{equation}
F_3^{(0)} = |\sf{1}{6}r(r^2-3s^2) + \sf{1}{4b}(r^2+s^2)+r(3bt+\sf{1}{4b^2})+
\sf12(t+\sf{1}{4b^3}) + i\,\sf{s}{6}(3r^2-s^2+18bt)|^2 \,,  
\label{F3}
\end{equation}
where we have set the parameters $\gamma_1=\gamma_2=\gamma_3=0$, for simplicity.
The approximate location of the peaks are obtained from solving the cubic
polynomial equations for $r$ and $s$, resulting from the real and 
imaginary parts of $F_3^{(0)}=0$ in \eqref{F3}.
The apporoximate peak locations are given by
\begin{align} 
&(r_p, s_p) \approx (\pm\sqrt{18b|t|}-\sf{2}{3b}, \,\, 0),  
\qquad &(r_p^c,s_p^c) \approx (-\sf{1}{6b}+\sf{29}{972b^4}\sf{1}{|t|}, \,\, 0),
\quad t \ll 0, \nonumber \\
&(r_p, s_p) \approx \big(\sf{5}{6b}+ \sf{83}{1944b^4}\sf{1}{t}, \,\, 
\pm(\sqrt{18bt}+\sf{25\sqrt{2}}{144b^{5/2}}\sf{1}{t^{1/2}})\big), \qquad
&(r_p^c,s_p^c) \approx (-\sf{1}{6b}-\sf{29}{972b^4}\sf{1}{t}, \, 0),
 \quad t \gg 0 \,,  
\label{F3zero}
\end{align} 
where $(r_p,s_p)$ denote the right (or left) peak locations
if $t \ll 0$, or the top (or bottom) peak locations if $t \gg 0$, and 
$(r_p^c,s_p^c)$ denotes the central peak locations.    
\begin{figure}[t!]
\begin{center}
\includegraphics[scale=0.5]{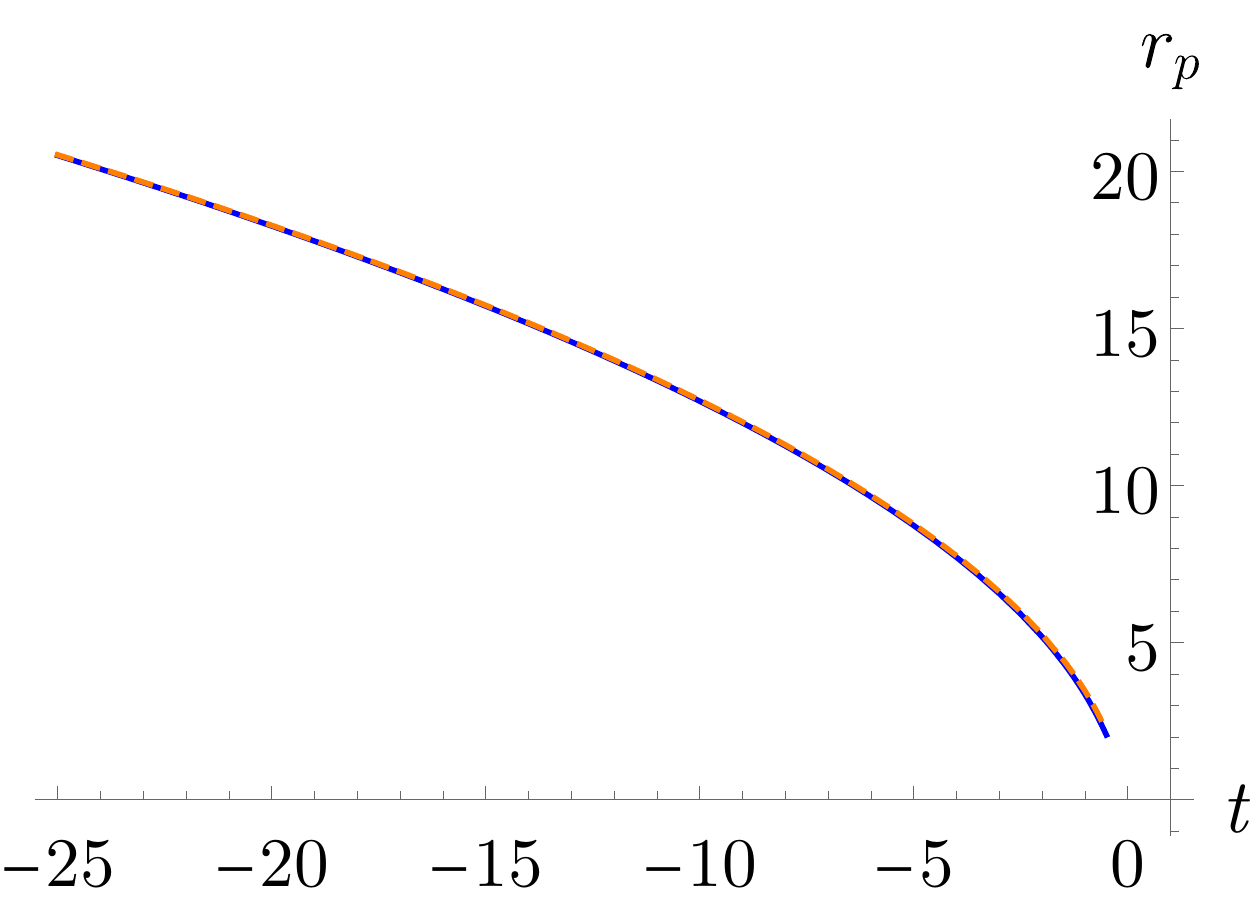} \quad \qquad
\includegraphics[scale=0.5]{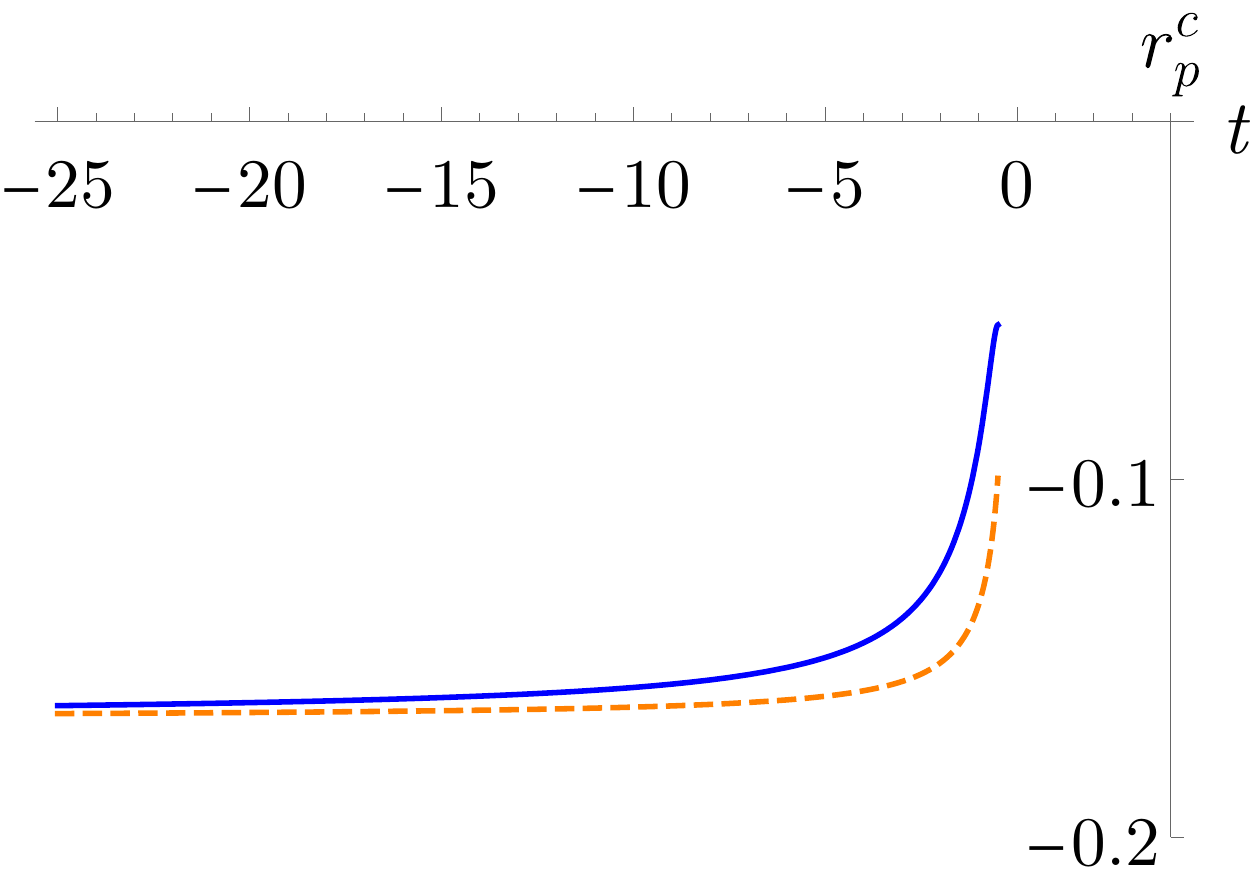}
\end{center}
\vspace{-0.2in}
\caption{Peak locations ($t \ll0$). Exact (blue), approximate (orange dashed).
Left panel:\, right peak, reflection across line $r_p=-\sf{2}{3b}$
gives the left peak location. Right panel:\, central peak.
KP parameters as in Figure~\ref{3lump}. }
\label{3lumppeaksN}
\end{figure}
The exact and approximate peak locations from \eqref{F3zero} 
plotted in Figures~\ref{3lumppeaksN} and \ref{3lumppeaksP}
are generally in good agreement for large $|t|$.
However, since the exact and approximate values differ by $O(|t|^{-1/2})$ 
like the $2$-lump case, the approximate value
of $r_p(t)$ for $t \gg 0$, obtained by minimizing $F_3$ gives
a better approximation $r_p(t)=\sf{5}{6b}-\sf{451}{6^5}\sf{1}{t}$
than the corresponding formula in \eqref{F3zero}.
The difference between the two approximations is in the $O(t^{-1})$ term.

When $t \ll 0$, the three peaks are well separated along the
$r$-axis with the central peak between $r= -\sf{1}{6b}$ and $r=0$.
Then with increasing time the left and right peaks approach each other 
and overlap with the central peak near the origin. As time evolves,
the $3$-lump solution then splits again into three distinct
peaks, the central peak remaining on the $r$-axis while 
the other two peaks separate from each other along the $s$-axis.
All three peak velocities have a small $r$-component which is the 
same for the top and bottom peaks while the central peak travels
a bit slower as can be seen from the right panel of Figure~\ref{3lump}.
Notice that the relative rate of attraction (repulsion)
between the left, right (top, bottom) is $\sqrt{18b}|t|^{-1/2}$
asymptotically as $|t| \to \infty$. The peak trajectories in the
$xy$-plane can be obtained using \eqref{rs} and \eqref{F3zero}
as in the $2$-lump case but are not shown here.
\begin{figure}[h!]
\begin{center}
\includegraphics[scale=0.42]{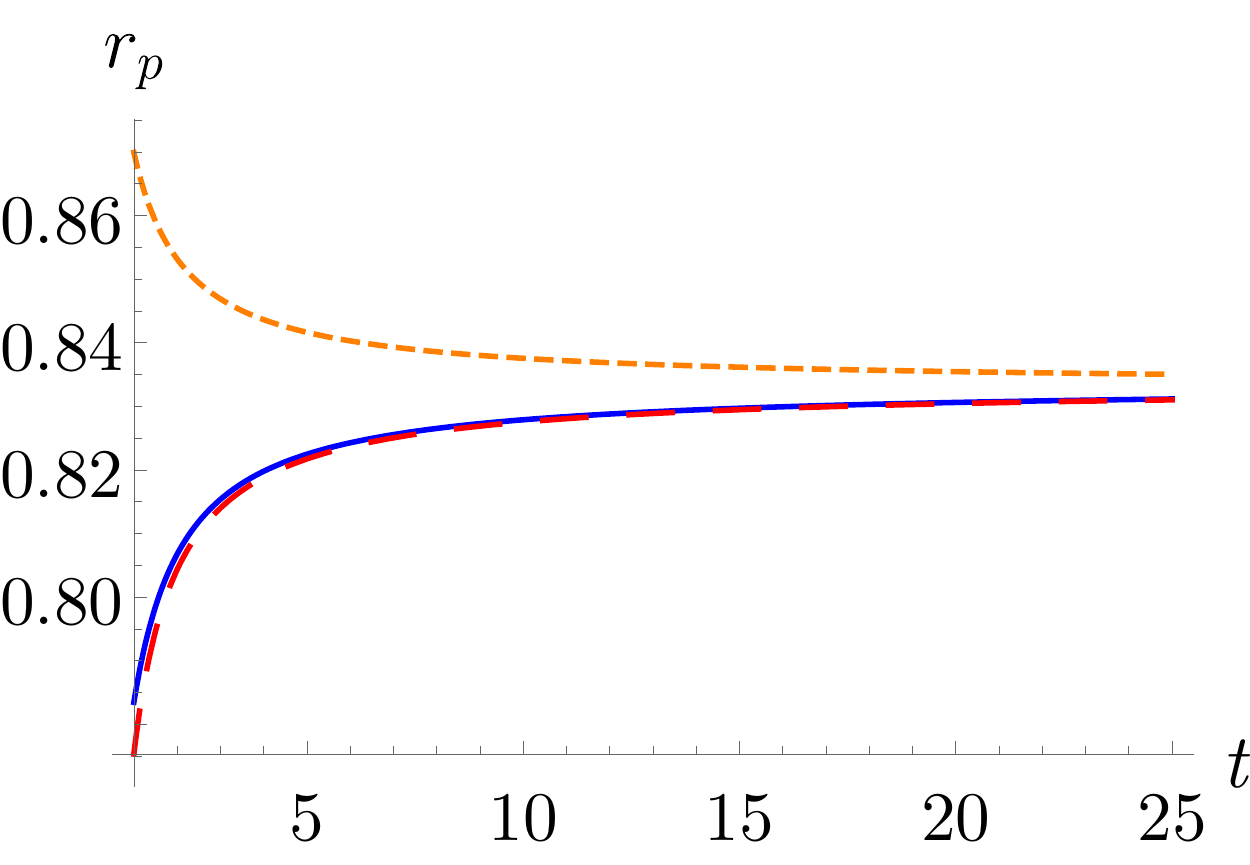} \quad 
\includegraphics[scale=0.42]{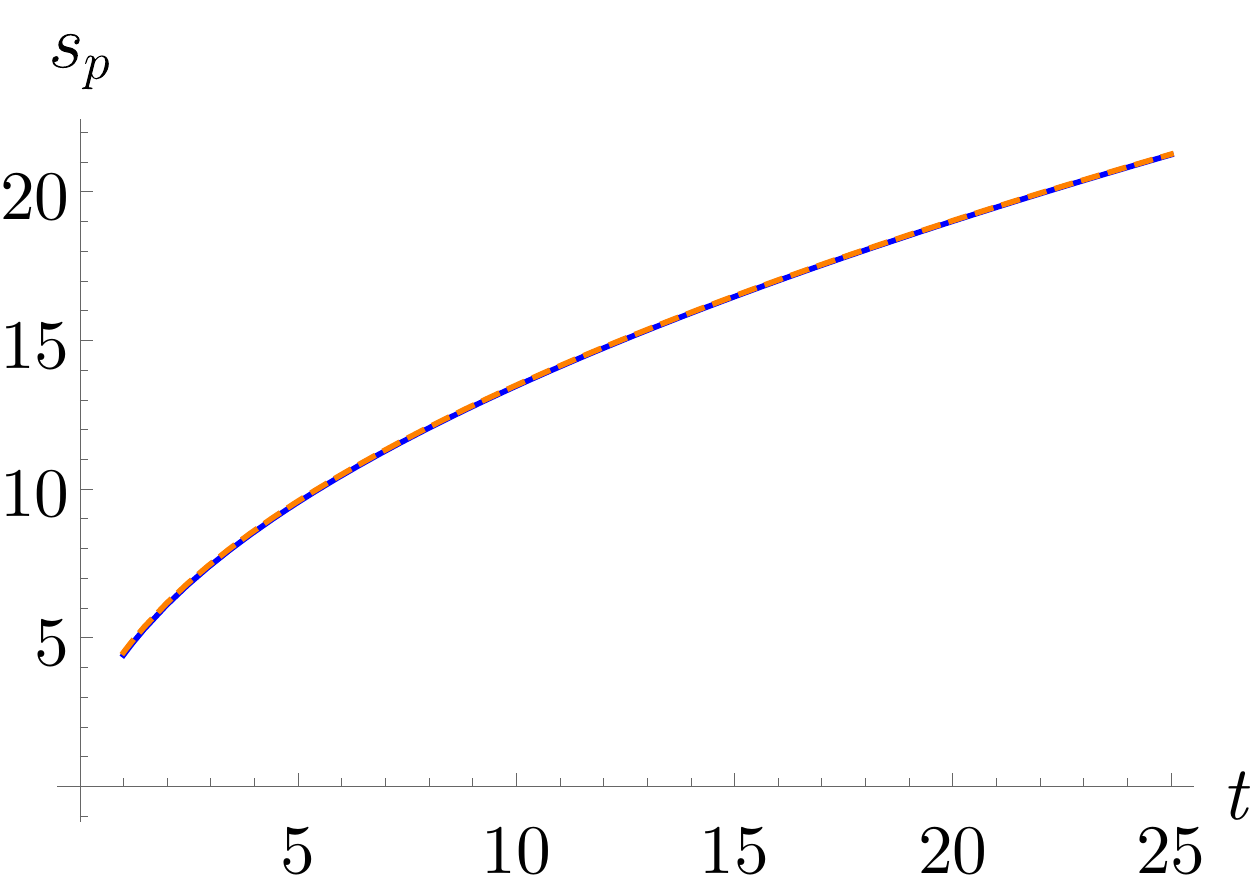} \quad 
\includegraphics[scale=0.42]{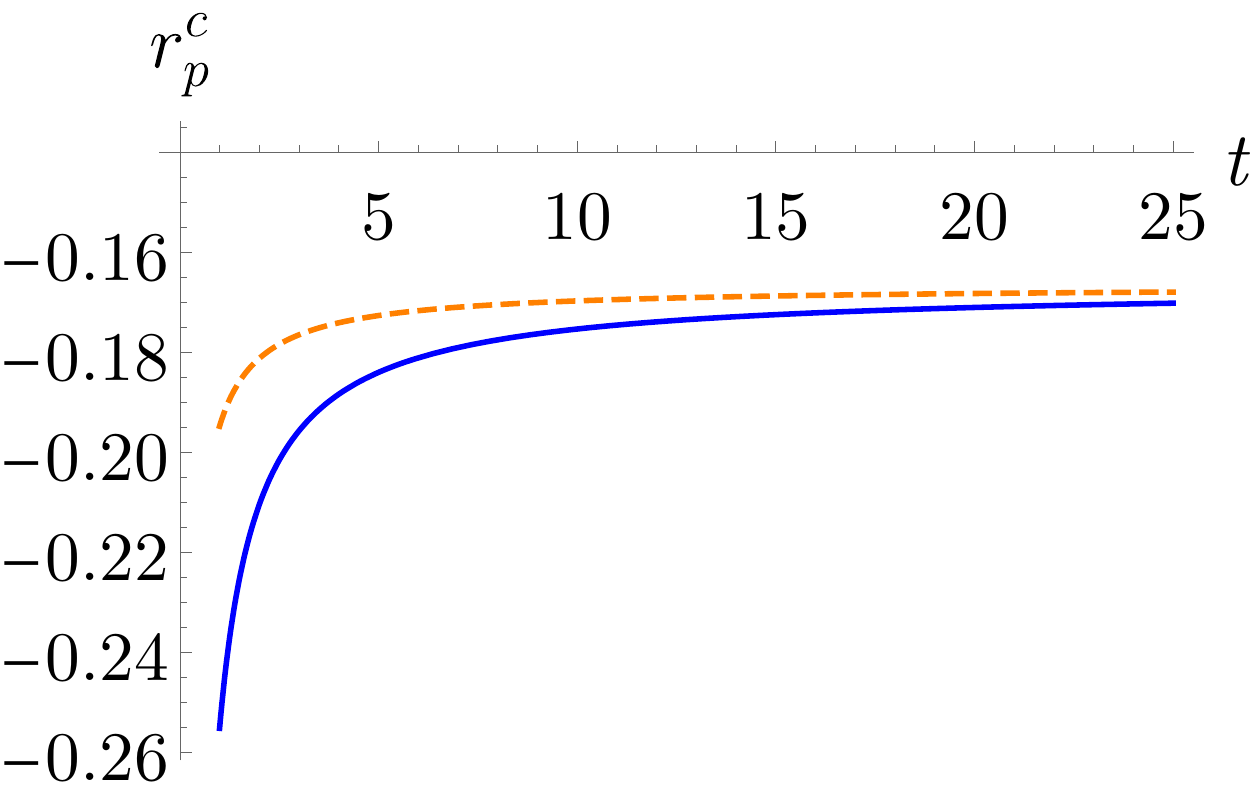}
\end{center}
\vspace{-0.2in}
\caption{Peak locations ($t \gg 0$). Exact -- solid blue, 
approximate from \eqref{F3zero} -- orange dashed. Left panel:\, $r_p(t)$ 
for top and bottom peaks, the red-dashed curve corresponds to minimum of $F_3$.
Middle panel:\, $s_p(t)$ (top peak) = $-s_p(t)$ (bottom peak).
Right panel:\, central peak.
KP parameters: same as in Figure~\ref{3lump}. }
\label{3lumppeaksP}
\end{figure}

The approximate peak heights are obtained by computing 
$u_{3p} = u_3(r_p,s_p)$ with $(r_p,s_p)$ from \eqref{F3zero}
like in the $2$-lump case. These are given by
\begin{align*}
&u_{3p}^{\pm} \approx 16b^2 \big(1 \mp \sf{1}{\sqrt{2b^3}}\sf{1}{|t|^{1/2}}
+\sf{5}{18b^3}\sf{1}{|t|}\big), \qquad
&u_{3p}^c \approx 16b^2\big(1+\sf{7}{36b^3}\sf{1}{|t|}\big), 
\quad t \ll 0,  \\
&u_{3p}^\pm \approx 16b^2 \big(1-\sf{5}{18b^3}\sf{1}{t} \big), \qquad 
&u_{3p}^c \approx 16b^2\big(1-\sf{7}{36b^3}\sf{1}{t} \big), 
\quad t \gg 0 \,,   
\end{align*}
where $u_{3p}^\pm, u_{3p}^c$ denote right or left (top or bottom)
and the central peak heights, respectively.
There are a few features to note from the above formulas which are
plotted in Figure~\ref{3lumppeakheights}.
First, as $t \to -\infty$ each peak
height approach the asymptotic value of $16b^2$. Secondly, both the 
left and the central peak grow 
as time evolves, whereas the right peak height decreases.
Moreover, the left peak grows at a faster rate than the central peak 
so that for $t \ll 0$, the peak heights are {\it ordered} with the left peak 
being the highest. After interaction, all the peaks grow at approximately
the same rate as each peak height approach $1$-lump value
of $16b^2$ as $t \to \infty$.

Finally, we remark that by using local coordinates 
$(r,s)=(r_p+h, s_p+k)$ near each peak, the polynomial $F_3$ reduces
to leading order in $t$, to a $1$-lump polynomial $F_1(r,s)$.
This implies that the $3$-lump solution $u_3$ is a superposition
of three $1$-lump solution as $|t| \to \infty$. In fact, it will be
shown in Section 4 that the $n$-lump solution is a superposition of $n$ 
$1$-lump solutions as $|t| \to \infty$. However, there are no 
asymptotic shifts in the peak centers before and after collision 
as in solitonic interactions e.g., the KdV solitons.
\begin{figure}[h!]
\begin{center}
\includegraphics[scale=0.42]{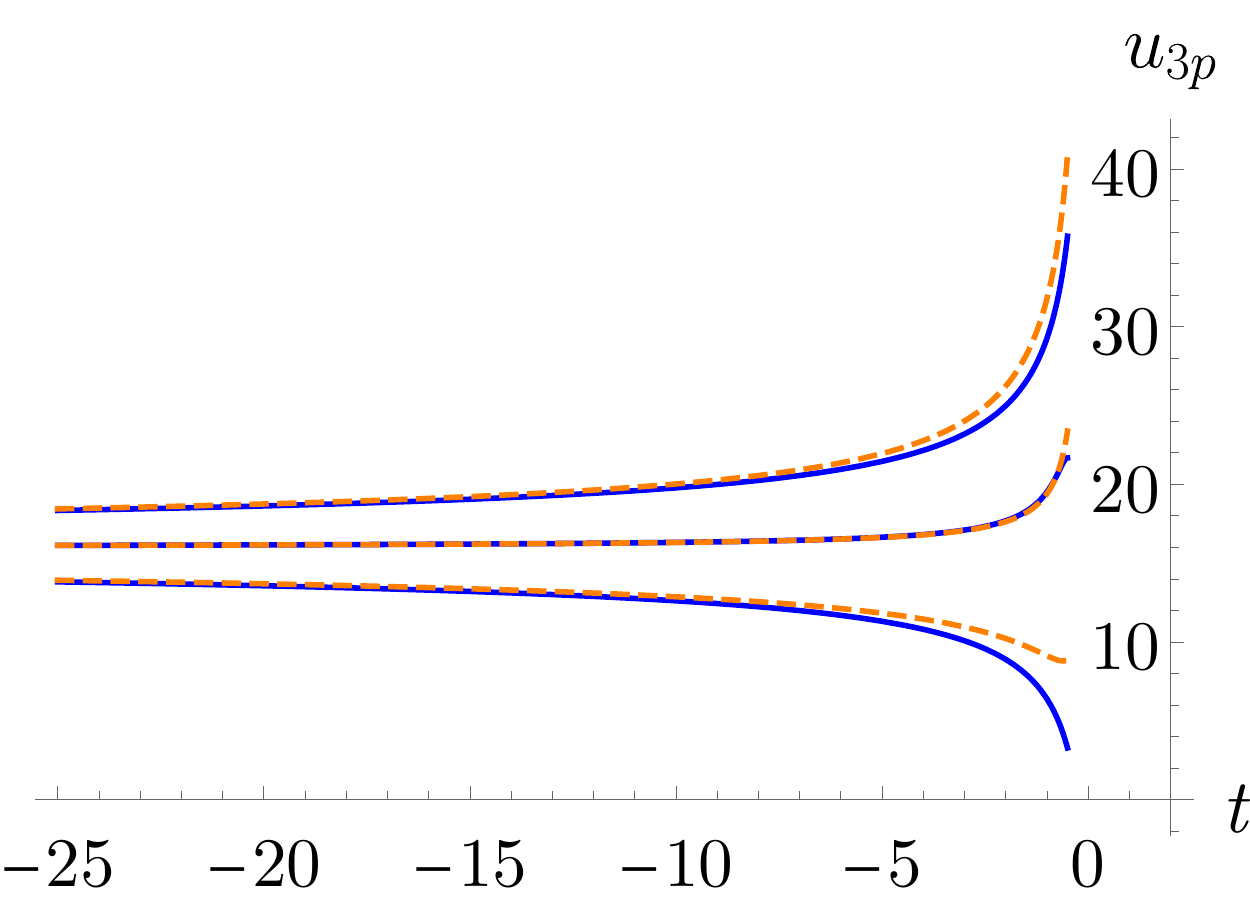} \quad
\includegraphics[scale=0.42]{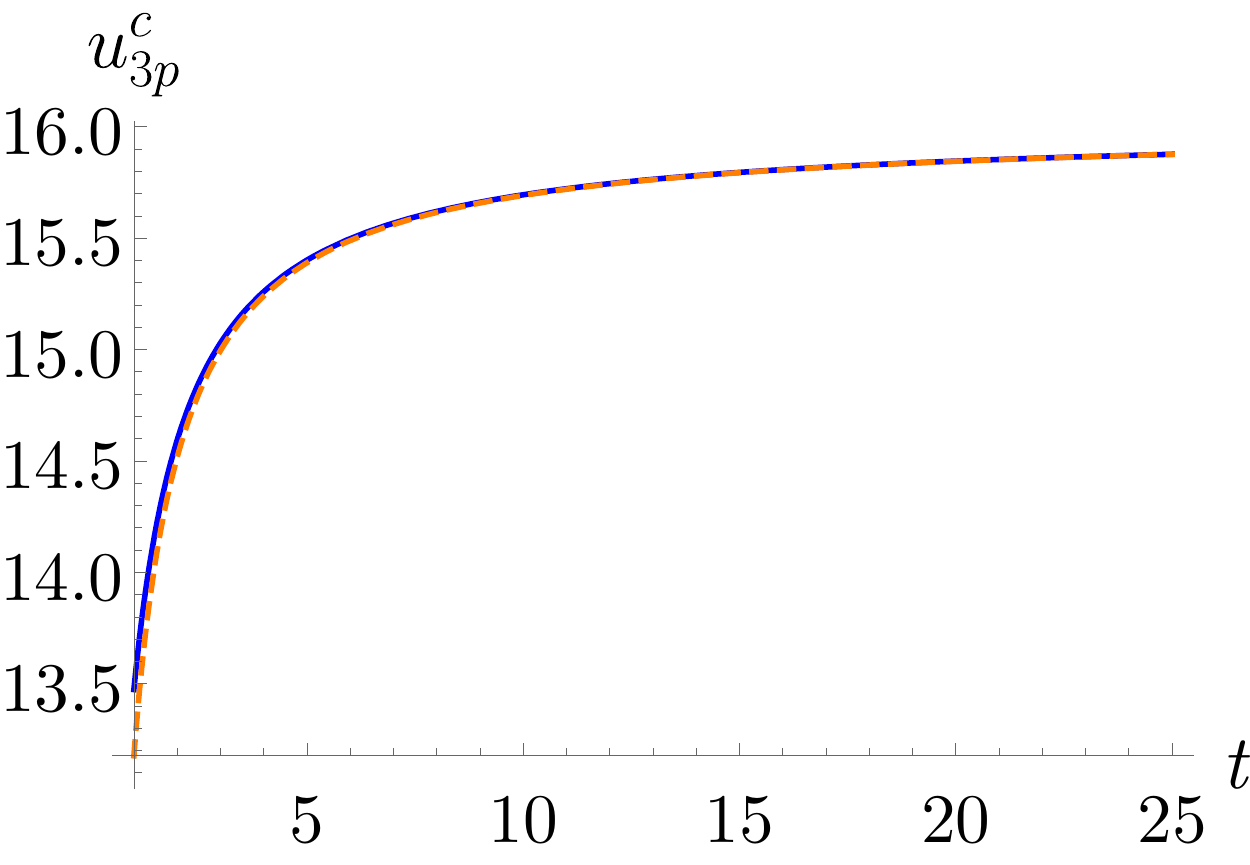} \quad
\includegraphics[scale=0.42]{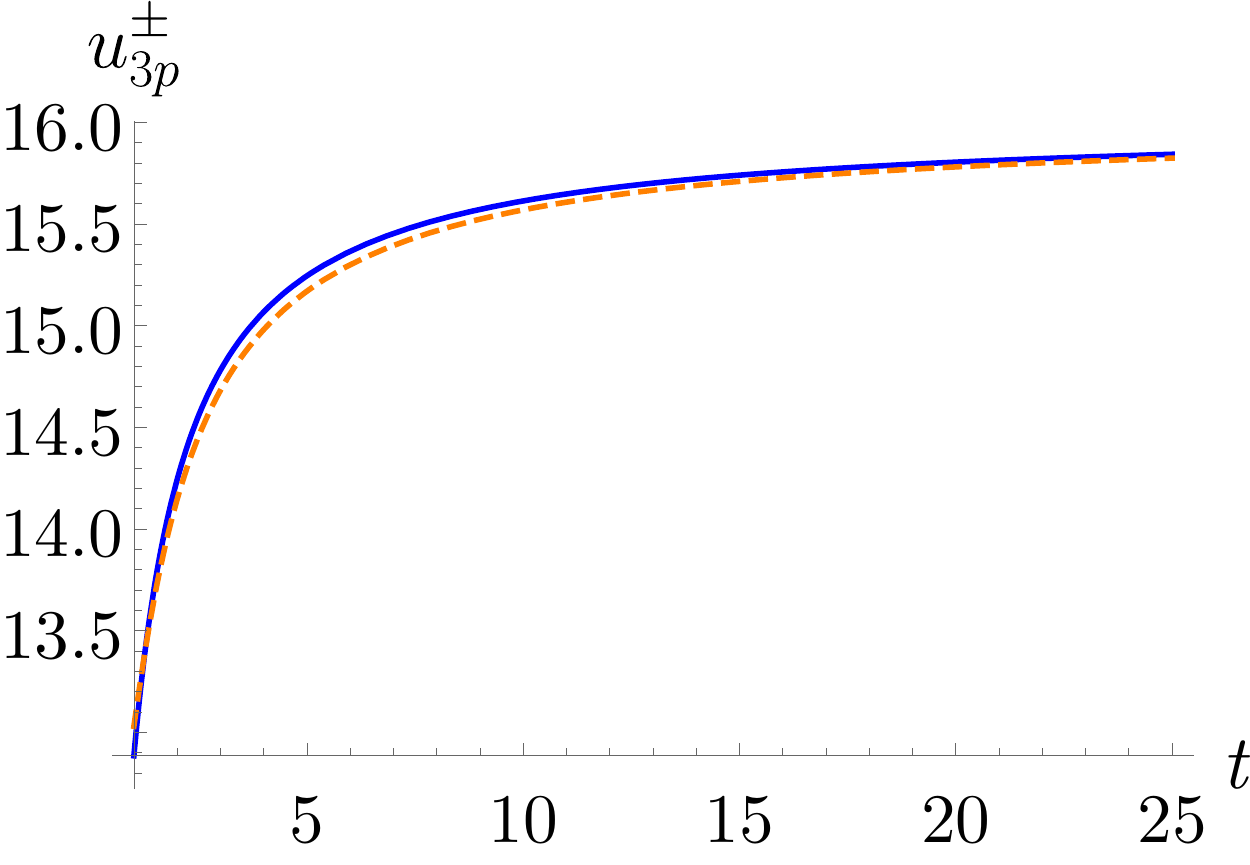}
\end{center}
\vspace{-0.2in}
\caption{Evolution of $3$-lump peak heights. Left panel:\, $u_{3p}^-$ (top plot), 
$u_{3p}^c$ (middle plot), and $u_{3p}^+$ (bottom plot) for $t \ll 0$. 
Middle panel:\, $u_{3p}^c$, right panel:\, $u_{3p}^\pm$ for $t \gg 0$.   
KP parameters: same as in Figure~\ref{3lump}.}
\label{3lumppeakheights}
\end{figure}
\section{$n$-lump asymptotics}
It is evident from the examples in Section 3 that the $n$-lump solution
for $n=2,3$ separates into $n$ distinct peaks whose heights approach the
$1$-lump peak height asymptotically as $|t| \to \infty$. Moreover, the
peak locations scale as $|t|^{1/2}$ and admit an asymptotic expansion
as $|t| \to \infty$ in the form of 
\[ z_j(t) := r_j(t) + is_j(t) \sim |t|^{1/2}\big(\xi_{j0} + \xi_{j1} \epsilon
+\xi_{j2} \epsilon^2 + \cdots\,, \big)\,, \qquad \epsilon = |t|^{-1/2} \,,\]     
where $z_j(t)$ is the $j^{\mathrm th}$ peak location, $j=1,2,\ldots,n$.
In this section, we will show that the features of the $n$-lump
solutions discussed in Section 3, also hold for 
any positive integer $n$.
The key point of the analysis is the fact that to leading order, the peak 
locations are given by $F_n^{(0)}(r,s,t)=0$ 
where $F_n^{(0)}$ is the $j=0$ term in the outer sum of \eqref{square}.
This fact will be justified in the Appendix.

Utilizing \eqref{propc} the sum inside $|\cdot|^2$ of $F_n^{(0)}$
can be expressed as a {\it single} generalized Schur polynomial 
\begin{equation}
 \sum_{l=0}^n(\sf{i}{2b})^lp_{n-l} = 
p_n(\theta_1+h_1, \theta_2+h_2, \cdots, \theta_n+h_n) := \tilde{p}_n
\label{pnt}
\end{equation}
by appropriately choosing the $h_j$'s such that
$p_j(h_1,\cdots,h_j) = (\sf{i}{2b})^j$. Thus $p_1(h_1)=h_1=\sf{i}{2b}$
so that from \eqref{thetars}, $\theta_1+h_1 = i(r+\sf{1}{2b}+is+\gamma_1)$
and all other $h_j$'s can be computed successively.   
Since $F_n^{(0)}=0$ implies $\tilde{p}_n(r,s,t)=0$, the approximate peak locations
are simply given by the zeros of the generalized Schur polynomial $\tilde{p}_n$,
to leading order in $|t|$. Recall from \eqref{thetars}, that the $t$-dependence 
in $\tilde{p}_n$ occurs via $\theta_2, \theta_3$ which are linear in $t$.
Then from \eqref{pn},
\[ \tilde{p}_n(r,s,t) \sim \sum_{m_1,m_2,m_3 \geq 0}
\frac{\theta_1^{m_1}\theta_2^{m_2}\theta_3^{m_3}}{m_1!m_2!m_3!}\,,
\qquad \text{with} \quad m_1+2m_2+3m_3=n \,,\]
when $|t| \gg 0$ since $\theta_j+h_j$ for $j>3$ are constants.
In order to find the terms with highest power of $|t|$ in  
$\tilde{p}_n$, one needs to maximize $m_2+m_3$ subject 
to $m_1+2m_2+3m_3 = n$. It is convenient to consider even and odd $n$ cases separately.
When $n=2m$, the maximizer is $(m_1,m_2,m_3)=(0,m,0)$ so that
the term with largest power of $|t|$
arises from $\theta_2^m/m!$ which is $O(|t|^m)$. When $n=2m+1$, there are 
two cases:\,
(i)\, $(m_1,m_2,m_3)=(1,m,0)$ and (ii)\, $(m_1,m_2,m_3)=(0,m-1,1)$ leading to the terms
$\sf{1}{m!}\theta_1\theta_2^m$ and $\sf{1}{(m-1)!}\theta_2^{m-1}\theta_3$ 
respectively, in $\tilde{p}_n$.
The highest power of $|t|$ in both those terms are $|t|^m$. 
Then the dominant balance required to solve the equation
\[\tilde{p}_n(r,s,t)=0 \]
asymptotically for $|t| \gg 0$ arises from the following two cases:  
\[ (a)\, \theta_1^{2m} \sim \theta_2^m, \,\, n=2m, \quad \text{and} \quad  
\theta_1^{2m+1} \sim \theta_1\theta_2^m, \,\, n=2m+1 \qquad
(b)\, \theta_1\theta_2^m \sim \theta_2^{m-1}\theta_3, \,\, n=2m+1 \,.\]
Case (a) gives the
scaling $\theta_1 \sim \theta_2^{1/2} \sim O(|t|^{1/2})$ as seen earlier for
the left/right (top/bottom) peak locations for the $2$- and $3$-lump solutions,
whereas case (b) implies $\theta_1 \sim O(1)$ as seen for the central peak locations
for the $3$-lump solution. Collecting all the dominant terms which are
$O(|t|^{n/2})$ from $\tilde{p}_n$ for case (a), yields
\begin{subequations}
\begin{equation}
\tilde{p}_n(r(t),s(t),t) \sim \sum_{m_1,m_2\geq 0}\frac{\theta_1^{m_1}\theta_2^{m_2}}{m_1!m_2!}
+O(|t|^{(n-1)/2}) \sim \,\, 
i^n\!\!\!\!\sum_{m_1,m_2\geq 0}\frac{z^{m_1}(3bt)^{m_2}}{m_1!m_2!} + O(|t|^{(n-1)/2}) \,,
\label{pnasymp-a}
\end{equation}
with $m_1+2m_2 = n$. The second expression above is obtained
by retaining the dominant terms in $\theta_1, \theta_2$ from \eqref{thetars} 
and defining $z(t) = r(t)+is(t)$ which 
is $O(|t|^{1/2})$. Similarly, the dominant terms corresponding
to case (b) are given by
\begin{equation}
\tilde{p}_n(r(t),s(t),t) \sim \frac{\tilde{\theta}_1\tilde{\theta}_2^m}{m!} + 
\frac{\tilde{\theta}_2^{m-1}\tilde{\theta}_3}{(m-1)!} + O(|t|^{m-1}) \sim 
i^n\left(\frac{(3bt)^mz}{m!} - \frac{(3bt)^{m-1}t}{(m-1)!}  
+ O(|t|^{m-1}) \right) \,, 
\label{pnasymp-b}
\end{equation}  
where $2m+1 = n, \, \tilde{\theta}_j = \theta_j+h_j, \, j=1,2,3$ and 
$z(t)=(r(t)-r_0)+\sf{1}{2b}+i(s(t)-s_0)$. Note that since 
$\theta_1 \sim O(1)$ in this case, all constant terms are retained 
leading to a different expression of $z(t)$ than in \eqref{pnasymp-a}.
The constants $r_0, s_0$ are determined by the 
parameters $\gamma_1, \gamma_2$, as usual.
\end{subequations}
\subsection{Asymptotic peak locations} 
First we consider case (b) above where $n$ is odd and $z(t) \sim O(1)$.
It follows immediately from \eqref{pnasymp-b} that $\tilde{p}_n=0$
implies $z(t) \sim \sf{m}{3b}+O(\sf{1}{|t|})$. Therefore, when $n=2m+1$, the
location of one of the peaks of the $n$-lump solution, to leading order,
is given by 
\begin{equation}
r(t) = r_0 + \sf{2m-3}{6b} +O(\sf{1}{|t|}), \quad s(t)=s_0, 
\qquad  m \geq 0 \,,
\label{centralpeak}
\end{equation}
for both $t \ll 0$ and $t \gg 0$.
This is the central peak located near the origin of the $rs$-plane since the 
remaining peak locations which scale as $O(|t|^{1/2})$ are located
either along the $r$- or $s$-axis as will be shown next.
Notice that the $1$-lump peak location from \eqref{u1} and the 
$3$-lump central peak location given by 
\eqref{F3zero} in Section 3 correspond to $m=0$ and $m=1$
respectively, with $(r_0,s_0)=(0,0)$.    

The remaining peak locations are obtained from \eqref{pnasymp-a}. In order
to obtain the leading order contribution, we put $z = |t|^{1/2}\xi$ in the
first term inside the sum of the second expression in \eqref{pnasymp-a} to
obtain 
\[\tilde{p}_n(r(t),s(t),t) \sim |t|^{n/2}h_n(\xi, \eta) + 
O(|t|^{(n-1)/2})\,, \]
where $h_n$ is a two-dimensional elementary
Schur polynomial also referred to as the heat polynomial~\cite{RW59} 
in $\xi$ and $\eta = \pm 3b$, the
$\pm$ sign depends on whether $t \gg 0$ or $t \ll 0$, respectively. 
Similar to the generalized Schur polynomials introduced in Section 2.1, the
heat polynomials have the generating function
\begin{equation}
\exp(\lambda \xi+\lambda^2\eta) = \sum_{n=0}^\infty h_n(\xi, \eta)\lambda^n\,,
\qquad h_n(\xi,\eta)= \sum_{p,q \geq 0}\frac{\xi^p\eta^q}{p!q!}, \quad p+2q=n\,,  
\qquad \eta = \pm 3b
\label{hn}
\end{equation}
and have the following useful properties
\begin{equation}
(a)\, \partial_\xi h_n = h_{n-1}, \qquad 
(b)\, (n+1)h_{n+1} = \xi h_n + 2\eta h_{n-1}, \qquad
(c)\, h_n(-\xi,\eta)=(-1)^nh_n(\xi,\eta) \quad \text{(parity)}\,,
\label{hnprop}
\end{equation}
for $n \geq 1$. The first few heat polynomials are given by 
$h_0=1, \, h_1 = \xi, \, h_2 = \sf{\xi^2}{2}+\eta, \cdots$.
Thus, by setting $\tilde{p}_n=0$ in \eqref{pnasymp-a}, the leading order
peak locations are given by $z_j(t)=r_j(t)+is_j(t)=|t|^{1/2}\xi_j$ where 
$\xi_j, \, j=1,2,\ldots,n$ are the complex roots of the heat poynomial 
$h_n$ regarded as a polynomial in complex variable $\xi$ only, with a fixed
real parameter $\eta$.
It remains to show that the roots are distinct. This is accomplished by 
Lemma~\ref{lemma1} below whose proof is elementary. 
\begin{lemma} \label{lemma1}
The heat polynomials $h_n(\xi,y)$ has a root at $\xi=0$ when
$n$ is odd. Moreover, $h_n$ has $n$ distinct real (in $\xi$),
roots if $\eta<0$, and $n$ distinct, pure imaginary (in $\xi$) 
roots if $\eta>0$.
\end{lemma}
\begin{proof}
The first part follows from property (c) in \eqref{hnprop} when $n$ is odd.

If $\eta < 0$, one proceeds by induction. The lemma clearly holds for the
base case $h_1(\xi, \eta)=\xi$. Assume it holds for $n \geq 1$, then
it follows from properties (a) and (b) of \eqref{hnprop} that 
$(n+1)h_{n+1} = 2\eta h_n'$ at each root of $h_n$, where $h_n'$ is the derivative
with respect to $\xi$. Since $h_n'$ changes sign at consecutive roots of $h_n$, so does
$h_{n+1}$. Therefore, $h_{n+1}$ has at least $n-1$ distinct roots, each between
two successive roots of $h_n$. Since for any $n$, $h_n \to \infty$ as $\xi \to \infty$ 
along the real $\xi$-axis, $h_n' > 0$ at the $n^{\mathrm th}$ root of $h_n$ where 
consequently, $h_{n+1} < 0$. Hence, $h_{n+1}$ has another root to the right
of the $n^{\mathrm th}$ root of $h_n$. Thus, $h_{n+1}$ has at least $n$
distinct roots. But since $h_{n+1}$ is a polynomial in $\xi$ with real coefficients
depending on the real parameter $\eta$, its roots are either real or 
complex conjugate pairs. Therefore, $h_{n+1}$ must have $n+1$ real, distinct roots.        

For $\eta > 0$, first put $\xi = i\hat{\xi}$ and note that $i^p = i^{n-2q}=i^n(-1)^q$
if $p+2q=n$. Then from \eqref{hn} it follows that $h_n(\xi,\eta) = 
i^nh_n(\hat{\xi},-\eta)$. Thus $h_n(\hat{\xi},-\eta)$ has
$n$ distinct, real roots when $\eta>0$. Therefore, $h_n$
has $n$ distinct, purely imaginary roots in $\xi$ including the root $\xi=0$
when $n$ is odd.
\end{proof}
In summary, Lemma~\ref{lemma1} and discussions preceding it demonstrate
that to leading order in $|t|$, the peaks of a $n$-lump solution are
separated into $n$ distinct peaks along the $r$-axis located at
\begin{subequations}
\begin{equation}
z_j(t) \sim |t|^{1/2}\big[(\xi_j, \, 0) + O(|t|^{-1/2})\big], 
\qquad  t \ll 0 \,, \\
\label{peak-r}
\end{equation}
and along the $s$-axis at 
\begin{equation}
z_j(t) \sim |t|^{1/2}\big[(0, \, \xi_j)+ O(t^{-1/2})\big], \qquad t \gg 0\,,
\label{peak-s}
\end{equation}
for each $j=1,2,\ldots,n$. Moreover, the $\xi_j$'s correspond to
the real roots of 
the heat polynomial $h_n(\xi, \eta)$ defined in \eqref{hn} with $\eta=-3b$.
\end{subequations}
When $n=2m+1$ is odd, the central peak is located on the $r$-axis,
given by \eqref{centralpeak}, and corresponds to the
root $\xi=0$ of $h_n(\xi,\eta)$ (see remark below).  
\paragraph{Remarks}
\begin{itemize}
\item[(a)] The heat polynomials $h_n(\xi,\eta)$ satisfies the heat equation
$h_{\xi\xi} = h_\eta$ in two-dimensions. They are closely related to the Hermite 
polynomials via $h_n(\xi,\eta) = \sf{(-\eta)^{n/2}}{n!}H_n(z), 
\, z=\xi(-4\eta)^{-1/2}$~\cite{RW59}. It is well known that the zeros of the 
Hermite polymials are real for real $z$. 
\item[(b)] Although we do not include here, it is possible to obtain the 
full asymptotic expansion for the roots $z(t)$ of $\tilde{p}_n =0$ by retaining the 
lower order terms in \eqref{pnasymp-a}. By setting $z(t)= |t|^{1/2}\xi$ 
in \eqref{pnasymp-a} and dividing by the leading power $|t|^{n/2}$, one finds 
that $\tilde{p}_n=0$ implies that $h_n(\xi, \eta) + 
O(\epsilon) = 0$ where $\epsilon = |t|^{-1/2}$. Substituting the expansion 
$\xi = \xi_0 + \epsilon \xi_1 + \ldots$ into this equation one can then calculate 
the $\xi_j$'s. The leading term $\xi_0$ satisfies $h_n(\xi_0, \eta)=0$,
and one can show that the $O(\epsilon)$ coefficient $\xi_1$ is real so that
$z(t)=\epsilon^{-1} \xi_0 + \xi_1 + O(\epsilon)$.
This then explains that for odd $n$, if $\xi_0=0$ then the central peak
location is real for $|t| \gg 0$ and is given by $z(t) = \xi_1+ O(\epsilon)$
as shown in \eqref{centralpeak}.
\end{itemize}
\subsection{Asymptotic peak heights}
The approximate peak heights are obtained by calculating
$u_j = u(z_j(t))$ where $z_j(t)=(r_j(t),s_j(t))$ are
the approximate peak locations given by $F_n^{(0)}=0$
and whose asymptotic expressions were given in Section 4.1.
We decompose the polynomial $F_n$ in \eqref{square} as $F_n=F_n^{(0)}+l$
where $l$ denotes all other square terms in the outer sum starting from $j=1$.
Then the $n$-lump solution at $z_j(t)$ can be expressed as
\[u_n = 2\partial_{xx}\ln(F_n^{(0)}+l) = 
2\frac{F_{nxx}^{(0)}+l_{xx}}{l}-
2\left(\frac{F_{nx}^{(0)}+l_x}{l}\right)^2 \,. \]
Writing $F_n^{(0)}=|A+iB|^2=A^2+B^2$, one obtains
$F_{nx}^{(0)} = 2(AA_x+BB_x) =0$ and $F_{nxx}^{(0)}=2(A_x^2+B_x^2)$ 
after using the fact that $A=B=0$ at $z=z_j(t)$. Furthermore, 
$A_x^2+B_x^2 = |A_x+iB_x|^2 = 
|\tilde{p}_{n-1}|^2$ which follows from
$\tilde{p}_{nx}=i\partial_{\theta_1}\tilde{p}_n$ and \eqref{propa}.
Therefore,
\begin{equation*}
u_n(z_j(t)) = 4\frac{|\tilde{p}_{n-1}|^2}{l}+
2\left(\frac{l_{xx}}{l}-\frac{l_x^2}{l^2}\right)\,.
\end{equation*}
We now proceed to estimate $\tilde{p}_{n-1}$ and $l$ in the above expression
for $u_n$ at the peak locations. Recall from \eqref{pnasymp-a} and \eqref{pnasymp-b}
that the dominant asymptotic behavior for $|t| \gg 0$ of the generalized 
Schur polynomials at the approximate peak locations are given by
\begin{subequations}
\begin{equation}
p_n(z_j(t)) \sim O(|t|^{n/2}), \qquad z_j(t) \sim O(|t|^{1/2}), \\
\label{pndominant-a}
\end{equation}
for any positive $n$, even or odd, and
\begin{equation}
p_n(z_j(t)) \sim p_{n-1}(z_j(t)) \sim O(|t|^{(n-1)/2}), 
\qquad z_j(t) \sim O(1)\,,
\label{pndominant-b}
\end{equation}
when $n$ is odd. The second estimate follows from the fact 
that both $p_n$ and $p_{n-1}$ has the same highest power of $|t|$ when
$n$ is odd. This is due to the weighted homogeneity property of the
$p_n$'s and can be deduced easily from \eqref{pn}.
\end{subequations}

The $j^{\mathrm th}$ term in the outer sum of \eqref{square} starts with $p_{n-j}$
in the sum inside $|\cdot|^2$. Thus it is evident from \eqref{pndominant-a}
and \eqref{pndominant-b} that the $j=1$ term will dominate in $l$ 
when $|t| \gg 0$. Using \eqref{propc} again, the sum inside $|\cdot|^2$ of 
the $j=1$ term can be expressed as a single generalized Schur polynomial with 
shifted arguments like in \eqref{pnt}  
\begin{equation}
\sf{i}{2b}\sum_{m=1}^n m (\sf{i}{2b})^{m-1}p_{n-m} =
\sf{i}{2b}p_{n-1}(\theta_1+c_1, \ldots,\theta_{n-1}+c_{n-1}) 
:=\sf{i}{2b}\hat{p}_{n-1} \,,  
\label{pnh}
\end{equation}
where the shifts $c_j$'s are determined
successively by $p_j(c_1,\ldots,c_j) =(j+1)(\sf{i}{2b})^j$ for
$j=1,\ldots,n-1$. From \eqref{pndominant-a} and \eqref{pndominant-b}
one finds that $\hat{p}_{n-1}(z_j(t)) \sim O(|t|^{(n-1)/2})$, and that
leads to
\begin{equation}
l(z_j(t)) \sim \sf{1}{4b^2}|\hat{p}_{n-1}|^2 + O(|t|^{n-2})\,,
\qquad |t| \gg 0 \,.
\label{lasymp}
\end{equation}
Furthermore, the derivatives $l_x, l_{xx}$ 
contain polynomials $p_{n-j}, \, j \geq 2$ and their
complex conjugates. Then it follows by using 
\eqref{pndominant-a} and \eqref{pndominant-b} again that  
both $(l_x/l)^2$ and $l_{xx}/l$ terms are $O(\sf{1}{|t|})$ at $z_j(t)$
when $|t| \gg 0$. Putting all these estimates together 
into the above expression $u_n(z_j(t))$ for the peak height, yields
\begin{equation}
u_n(z_j(t)) \sim 16b^2\frac{|\tilde{p}_{n-1}|^2}{|\hat{p}_{n-1}|^2}
+O(\sf{1}{|t|})  \,.
\label{unasymp}
\end{equation}
The generalized Schur polynomials $\tilde{p}_{n-1}$ and $\hat{p}_{n-1}$ in 
\eqref{pnt} and \eqref{pnh} are of the same
degree but only with shifted arguments. Hence, one can be expressed in 
terms of the other by using \eqref{propc} so that
$\tilde{p}_{n-1} = \hat{p}_{n-1} + O(|t|^{(n-2)/2})$ from
\eqref{pndominant-a} or \eqref{pndominant-b}.  
Consequently, $\tilde{p}_{n-1}/\hat{p}_{n-1} \sim 1 + O(\sf{1}{|t|^{1/2}})$ at
each peak location $z_j(t)$. Then from \eqref{unasymp}, it follows that
asymptotically as $|t| \to \infty$, each peak height $u_n(z_j(t)) \sim 16b^2$
for $j=1,\ldots,n$, as claimed at the beginning of this section.

An interesting result follows from \eqref{unasymp} regarding the ordering
of the peak heights for $t \ll 0$, that was observed for the $3$-lump
solution in Section 3.3 (see left panel of Figure~\ref{3lumppeakheights}).
It was shown in Section 4.1 that 
$\tilde{p}_n \sim |t|^{n/2}h_n(\xi, \eta)$ for a positive integer $n$.
Then, $\tilde{p}_{n-1}(z(t)) \sim |t|^{(n-1)/2}h_{n-1}(\xi, \eta)$ and 
$\hat{p}_{n-1}(z(t)) \sim |t|^{(n-1)/2}h_{n-1}(\xi+\epsilon_1, \eta)$ where
the shift $\epsilon_1$ is obtained as follows. From \eqref{pnt} and \eqref{pnh}
one can see that $\theta_1$ is shifted by different amounts namely,
$h_1=\sf{i}{2b}$ and $c_1=\sf{i}{b}$ in $\tilde{p}_n$ and $\hat{p}_n$, respectively.
This induces an additional shift $z(t) \to z(t)+\sf{1}{2b}$ in the 
argument of $\hat{p}_{n-1}$. Hence, after the scaling $z(t) = |t|^{1/2}\xi$, 
one obtains $\xi \to \xi+\epsilon_1$ where 
$\epsilon_1=\sf{\epsilon}{2b}, \,\, \epsilon = |t|^{-1/2}$.
Then for $t \ll 0$, \eqref{unasymp} can be re-expressed as
\[ u_n(z_j(t)) \sim 16b^2
\frac{h_{n-1}^2(\xi_j, -3b)}{h_{n-1}^2(\xi_j+\epsilon_1, -3b)} \,,\]
where $\xi_j, \, j=1,\ldots,n$ are $n$ real roots of the polynomial
$h_n(\xi,-3b)$. Expanding $h_{n-1}(\xi_j+\epsilon_1)$ and using relations 
(a), (b) from \eqref{hnprop}, and the fact that $h_n(\xi_j)=0$, yields 
\[h_{n-1}(\xi_j+\sf{\epsilon}{2b}) = h_{n-1}(\xi_j) + 
\sf{\epsilon}{2b} h_{n-2}(\xi_j) + O(\epsilon^2) = 
h_{n-1}(\xi_j)\left(1+\frac{\epsilon \xi_j}{12b^2}\right) 
+ O(\epsilon^2)\,.  \]
Substituting this into the above expression for $u_n(z_j(t))$ gives
\[ u_n(z_j(t)) \sim \frac{16b^2}{\big(1+\sf{\xi_j\epsilon}{12b^2}\big)^2}
+O(\epsilon^2) \,, \]
which demonstrates that the $n$-lump peak heights are in descending order from
left to right for $t \ll 0$. Note that for the central peak, when
$\xi_j=0$, it is necessary to compute the  $O(\epsilon^2)$ contribution
to $u_n(z_j(t))$.  
\subsection{Asymptotic form of solution and a dynamical system}
This subsection is devoted to certain topics arising from
the asymptotic analysis developed in Sections 4.1 and 4.2. 
First we deduce the long time asymptotics of the general $n$-lump
solution of KPI; then we derive a dynamical system for the zeros 
$\xi_j$ of the heat polynomials $h_n(\xi, \eta)$.

\subsubsection{Asymptotic form of $n$-lump solution}
We begin by deriving the local form of the polynomial $F_n$ near 
each peak location $z_j(t) \, j=1,\ldots,n$. As before, we denote 
$z:=(r,s), \, z_j(t) = (r_j(t), s_j(t))$, and put
$z = z_j(t)+(h, k)$ in $F_n = F_n^{(0)}+l = |\tilde{p}_n|^2+l$.
Then we expand in $h, k$ and retain the leading order terms in $|t|$
in $F_n$. Using \eqref{rs}, \eqref{propa} and the fact that 
$\tilde{p}_n(z_j(t))=0$, one obtains
\[\tilde{p}_n(z) = (ih-k)\tilde{p}_{n-1}(z_j(t))+
\sf{i}{2b}k\tilde{p}_{n-2}(z_j(t)) +Q(h,k)\,, \]
where $Q(h,k)$ consists of quadratic and higher powers in $h,k$
whose coefficients are expressed in terms of $\tilde{p}_{n-l}, \, l \geq 2$.
Employing \eqref{pndominant-a} and \eqref{pndominant-b} to estimate
the dominant behavior of $p_l(z_j(t))$ in $F_n^{(0)}$ then yields,
\[|\tilde{p}_n|^2 (z) \sim 
(h^2+k^2)|\tilde{p}_{n-1}|^2(z_j(t))+O(|t|^{n-2})\,.\]
It was shown in Section 4.2 that the dominant term in $l$ arises
from the $j=1$ term of \eqref{square} and its leading order contribution
for $|t| \gg 0$ is
\[ l(z) \sim \sf{1}{4b^2}|\hat{p}_{n-1}|^2(z_j(t)) + O(|t|^{n-2})
\sim \sf{1}{4b^2}|\tilde{p}_{n-1}|^2(z_j(t)) + O(|t|^{n-2}) \,.\]
Thus the asymptotic form of $F_n$ is given by
\begin{equation}
F_n (h,k; z_j(t)) \sim |\tilde{p}_{n-1}|^2(z_j(t)) \left(h^2+k^2+\sf{1}{4b^2}
+ O(\sf{1}{|t|}) \right) \,. 
\label{Fnasymp}
\end{equation}
Finally, substituting \eqref{Fnasymp} for each peak
location $(z_j(t)), \, j=1,\ldots, n$ in \eqref{Fn} gives
\[u_n(h,k;z_j(t)) \sim u_1(h,k) + O(\sf{1}{|t|}) \,.\]
Thus we have shown that asymptotically 
as $|t| \to \infty$, $u_n$ is a superposition of $n$ distinct
$1$-lump solutions whose peaks are located at $z_j(t), \, j=1,\ldots, n$. 
Furthermore, since $\int\!\!\int_{\mathbb{R}^2}u_n^2$ is time-invariant
it follows immediately that $\int\!\!\int_{\mathbb{R}^2}u_n^2 = 
n\int\!\!\int_{\mathbb{R}^2}u_1^2 + O(\sf{1}{|t|}) 
= n\int\!\!\int_{\mathbb{R}^2}u_1^2 = n(16\pi b)$.

\subsubsection{Dynamical system}
The heat polynomial $h_n(\xi,\eta)$ introduced in Section 4.1 
is in fact a rescaled version of the polynomial
\begin{subequations}
\begin{equation}
h_n(z,\tau) = \sum_{m_1,m_2 \geq 0}\frac{z^{m_1}\tau^{m_2}}{m_1!m_2!}, 
\qquad m_1+2m_2=n, \qquad \tau = 3bt \,, 
\label{heat-a}
\end{equation}
which is the dominant term in the second expression of the
asymptotic expansion in \eqref{pnasymp-a}. It is easy to see that
\begin{equation}
h_n(z,\tau) = |t|^{n/2}h_n(\xi, \eta)\,, \qquad 
z = |t|^{1/2}\xi, \qquad t=\tau/3b \,,
\label{heat-b}
\end{equation}
where $\eta=\pm 3b$ if $t>0$ or $t<0$, respectively. 
For each positive integer $n$,
$h_n(z,\tau)$ is a polynomial solution of the heat equation 
$u_\tau=u_{zz}$ with initial condition
$u(z,0)=h_n(0) = \sf{z^n}{n!}$ whose solution can be expressed as
\[ u(z,\tau) = \sf{1}{n!}\exp(\tau \partial_z^2)(z^n) = 
\sf{1}{n!}\sum_{j=0}^\infty\frac{\tau^j}{j!}\partial_{z}^{2j}(z^n) = h_n(z,\tau)\,. \]
\end{subequations}
The goal here is to investigate the dynamics of the zeros $z_k(\tau)$ of 
$h_n(z,\tau)=0$ with respect to the evolution variable $\tau$. This is
accomplished by differentiating
implicitly the equation $h_n(z_k(\tau),\tau)=0$, which yields
\[h_{nz}z_k'(\tau) + h_{n\tau} = h_{nz}z_k'(\tau) + h_{nzz} =0 \,, \]
after using the heat equation.
Now from Lemma~\ref{lemma1} and \eqref{heat-b} it is clear that
$h_n(z,\tau)=0$ has $n$ distinct roots. Hence, it can be expressed
as $h_n(z,\tau) = \sf{1}{n!}(z-z_1(\tau))(z-z_2(\tau))\cdots(z-z_n(\tau))$ 
so that the logarithmic derivative 
\[h_{nz}/h_n = \sum_{j=1}^n\frac{1}{z-z_j(\tau)}\,.\]
Moving the $k^{\mathrm th}$ term from the above sum to the left
and letting $z \to z_k(\tau)$, one recovers after using L'hopital's rule,
a system of ordinary differential equations for the zeros 
$z_k(\tau), \, k=1,\dots,n$, namely  
\begin{equation}
z_k'(\tau) = - \lim_{z \to z_k}\frac{h_{nzz}}{h_{nz}} = 
\sum_{j\neq k}\frac{2}{z_j-z_k} \,, \qquad z_k(0)=0 \,.
\label{ds}
\end{equation}
The dynamical system \eqref{ds} is one of the simplest examples describing
the motion of zeros of the solution to a partial differential equation. 
This system was studied
in different contexts such as moving poles of the complex Burger's 
equation~\cite{CC77,C78}, vortex dynamics~\cite{A83}, and more recently in proving 
certain conjecture related to the Riemann hypothesis~\cite{CSV94,RT20}. 
Equation \eqref{ds} can be expressed as a gradient
flow of a logarithmic potential
\[ z_k'(\tau) = -\frac{\partial V}{\partial z_k}\,, \qquad 
V=\sum_{\myatop{i,j}{i\neq j}}\log|z_i-z_j| \,,\]   
whose evolution can be computed as follows
\[
V'(\tau) = \sum_{k=1}^nz_k'(\tau)\frac{\partial V}{\partial z_k} =
-\sum_{k=1}^n\left(\frac{\partial V}{\partial z_k}\right)^2 = 
-\sum_{k=1}^n \sum_{\myatop{i,j}{i \neq k, j \neq k}}
\frac{4}{(z_i-z_k)(z_j-z_k)} 
=-\sum_{\myatop{i,j}{i \neq j}}\frac{4}{(z_i-z_j)^2} \,,  
\]
where the cross-terms vanish in the last inequality. Consequently,
the acceleration $z_k''(\tau)$ satisfies
\[z_k''(\tau) = -\frac{\partial V'(\tau)}{\partial z_k} = 
 \sum_{i \neq k}\frac{16}{(z_i-z_k)^3} \,,\]
which is the well-known Calogero-Moser system~\cite{C71}.
Thus, \eqref{ds} implies that the $z_k$'s satisfy the Calogero-Moser 
system although the converse is not necessarily true.

A complete description of the zeros of the polynomial $h_n(z,\tau)$ 
is of course readily available from the transformation given in \eqref{heat-b}
in terms of the polynomials $h_n(\xi, \eta)$ and particularly,
Lemma~\ref{lemma1} of Section 4.1.
For $\tau<0$, the zeros $z_k(\tau) = r_k(\tau)+is_k(\tau)=
(\sf{|\tau|}{3b})^{1/2}\xi_k$ are real, distinct, 
and symmetrically located about the origin along the $r$-axis, they all coalesce at 
$\tau=0$ and emerge along the $s$-axis as purely imaginary and complex conjugate 
roots for $\tau>0$. Nonetheless, one gains further insights into their dynamics 
from analyzing the system \eqref{ds}. Fix a $\tau_0<0$ and consider the interval
$I=[z_k(\tau_0), z_{k+1}(\tau_0)]$. Clearly the first derivative $h_{nz}$ has
opposite signs at the zeros $z_k(\tau_0)$, and $z_{k+1}(\tau_0)$. Now suppose that 
the derivative $h_{nz}$ is positive at $z_k(\tau_0)$ and negative at $z_{k+1}(\tau_0)$
so that the concavity $h_{nzz} < 0$ on $I$, then from \eqref{ds} it follows
that $z_k'(\tau_0) > 0$ whereas $z_{k+1}'(\tau_0) < 0$, which means that the
zeros behave like point particles attracting each other for $\tau <0$. 
The conclusion remains
the same if the sign of $h_{nz}$ is reversed at the end points of the interval $I$.
In this situation, the logarithmic potential $V(\tau)$ is a 
monotonically decreasing function of $\tau$
since $V'(\tau)<0$, and reaches a singularity 
($V \to -\infty$) when
all the zeros collide at $\tau=0$. An opposite scenario is observed for $\tau>0$
by replacing $z_j(\tau) = i\zeta_j(\tau)$ in \eqref{ds}, and $\zeta_j(\tau)$ is now real. 
The zeros now repel from each other and $V(\tau)$ increases monotonically
in $\tau$ as the zeros move further apart from each other.     
From the viewpoint of point particle dynamics the force of attraction (repulsion)
is given by the Calogero Moser system above depending on if the $z_k$'s are 
real (pure imaginary) for $\tau<0$ ($\tau>0$).

In conclusion, we have demonstrated that the $n$-lump solution of the KPI
equation splits into $n$ distinct peaks for $|t| \gg 0$ whose locations
evolve as $n$ interacting point particles as given by \eqref{ds}. However,
the dynamics only describe the leading order peak locations for large $|t|$
and can not be used to infer the evolution when $|t| \sim O(1)$ or
smaller.

\paragraph{Remarks}
\begin{itemize}
\item[(a)] Substituting $z_k(\tau) = (\sf{|\tau|}{3b})^{1/2}\xi_k$ 
in \eqref{ds} gives the following relations among the roots $\xi_k$
of $h_n(\xi,\eta)$
\[ x_k = \mp \sum_{j \neq k} \frac{1}{x_j - x_k}\,, 
\qquad \xi_k = 2\sqrt{3b}\,x_k\,, \quad 1 \leq k \leq n \,,\]
where the $\mp$ sign depends respectively, on whether
the the roots are real or pure imaginary.
These are called the Stieltjes relations~\cite{S85a,S85b} for the zeros 
of the Hermite polynomials. 
The Stieltjes relations can be used to determine the roots $x_k$
exactly.

\item[(b)] The dynamical system \eqref{ds} was also obtained 
in~\cite{M79,P94} by considering zeros of wronskian $\tau$-functions 
of KPI. This wronskian can in fact be identified with 
the generalized Schur polynomial $\tilde{p}_n$ in \eqref{pnt}.
In this paper, it was explicitly shown that the zeros of
$\tilde{p}_n$ are indeed {\it distinct} for $|t| \gg 0$,
and dynamical system \eqref{ds} was used to explain the anomalous 
nature of the $n$-lump interaction.
\end{itemize}

\section{$n$-lump eigenfunction}
The KPI equation admits a Lax pair which is the following linear
system of equations
\begin{equation}
i\psi_{y} = \psi_{xx}+u\psi, \quad \qquad 
\psi_t = -\psi_{xxx}-\sf{3}{2}u\psi_x-\sf{3}{4}(u_x+iw)\psi, \qquad w_x=u_y\,,  
\label{lax}
\end{equation} 
where the common solution $\psi(x,y,t)$ is usually called the 
eigenfunction and $u(x,y,t)$ is the KPI solution. It follows
from the integrability condition $\psi_{ty}=\psi_{yt}$ of \eqref{lax} 
that $u(x,y,t)$ satisfies the KPI equation \eqref{kp}.  
It is significant to note that the first equation in \eqref{lax} 
is the well-known non-stationary Schr\"odinger equation (NSE) with 
$u(x,y,t)$ regarded as the potential function parametrized by $t$. 
Thus, any solution of KPI equation at a fixed $t$ is a potential
function and in particular, the $n$-lump solution $u_n(x,y,t)$
belongs to the class of rational potentials for the NSE.
\subsection{Binary Darboux transformation} 
The aim of this section is to derive an explicit form for the
eigenfunction $\psi$ of \eqref{lax} associated with the $n$-lump
potential $u_n$ by employing binary Darboux transformation~\cite{MS91}. 
Such eigenfunctions were found via the IST method~\cite{AV97,VA99} 
and the $\bar{\partial}$-dressing method~\cite{Dub99} in earlier studies 
which outlined an explicit scheme to
construct $\psi$ and $u_n$, and which works well for small $n$ values.
However, the scheme becomes unwieldy to obtain general expressions
for $\psi$ for aribtrary $n$ -- a task that is accomplished by the
algebraic method presented here.

Let us first start with solutions of \eqref{lax} when $u=0$.
In this case, \eqref{lax} coincides with \eqref{phi}. Then it is
possible to verify the following lemma by direct calculation.
\begin{lemma}
\label{lemma2}
Suppose $\psi_0$ is any solution of \eqref{phi},
and $\phi_n(k_0)$ defined by \eqref{sp} is another particular solution 
of \eqref{phi} with a fixed complex parameter $k_0=a+ib$.   
Then 
\[ \hat{\psi} = -\frac{M(\psi_0,\bar{\phi}_n)}{\bar{\phi}_n}\,, \qquad 
M(\psi_0,\bar{\phi}_n)= 
\int_x^\infty\psi_0\bar{\phi}_n\,dx \]
is a solution of \eqref{lax} with $u=\hat{u}:=2\ln(\bar{\phi}_n)_{xx}$. 
\end{lemma}
Recall that $\bar{\phi}_n$ is the complex conjugate of $\phi_n$.
The transformation $(\psi_0, u=0) \to (\hat{\psi}, \hat{u})$
is called an adjoint Darboux transformation.    
Next we introduce an ordinary Darboux transformation.
\begin{lemma}
Suppose $(\hat{\psi}, \hat{u})$ satisfies \eqref{lax} as in Lemma~\ref{lemma2}
and $\hat{\psi}_n = -M(\phi_n,\bar{\phi}_n)/\bar{\phi}_n$ is a particular
solution \eqref{lax} with $u=\hat{u}$. Then 
$\psi_n = \hat{\psi}_x - \hat{\psi}(\ln \hat{\psi}_n)_x$ is a solution of
\eqref{lax} with $u = \hat{u}+2\ln(\hat{\psi}_n)_{xx}=u_n$ where
$u_n$ is the $n$-lump solution.
\label{lemma3}
\end{lemma}
By replacing $\psi_0$ by $\phi_n$ in the expression for $\hat{\psi}$
of Lemma~\ref{lemma2}, it should be clear that $\hat{\psi}_n$ solves
\eqref{lax} with $u=\hat{u}$ as claimed in Lemma~\ref{lemma3}.
It is easy to verify from the definitions of $\hat{u}, \hat{\psi}_n$,
\eqref{taun} and \eqref{Fn} in Section 2.2 that 
$\hat{u}+2\ln(\hat{\psi}_n)_{xx}
=2\ln(\bar{\phi}_n)_{xx}+ 2\ln(\hat{\psi}_n)_{xx}=
2\partial_x^2\ln M(\phi_n,\bar{\phi}_n)=u_n$.    
Thus $(\hat{\psi}, \hat{u}) \to (\psi_n,u_n)$ is the ordinary Darboux
transformation. The composition of the adjoint and ordinary Darboux
transformations gives 
\[(\psi_0, u=0) \overset{\phi_n,\bar{\phi}_n}{\longrightarrow} 
(\psi_n,u_n)\]
which is called the binary Darboux transformation that is
triggered by the solutions $\phi_n, \bar{\phi}_n$ of 
\eqref{phi} or \eqref{lax} with $u=0$. 
To further clarify this transformation, note that 
from the definitions of $\hat{\psi}, \, \hat{\psi}_n$ in 
Lemmas 2 and 3, one obtains  
\[ \psi_n = \hat{\psi}\partial_x\ln(\hat{\psi}/\hat{\psi}_n)
= -\frac{M(\psi_0,\bar{\phi}_n)}{\bar{\phi}_n}
\partial_x\ln\left(\frac{M(\psi_0,\bar{\phi}_n)}{M(\phi_n,\bar{\phi}_n)}\right) \,.
\]
Then after using the fact that $M(f,g)_x=-fg$ in the above expression, yields the
following expression for the $n$-lump eigenfunction and the solution   
\begin{equation}
\psi_n = \psi_0 - \frac{M(\psi_0,\bar{\phi}_n)}{M(\phi_n,\bar{\phi}_n)}\,\phi_n\,,
\quad \qquad u_n = 2\partial_x^2\ln M(\phi_n,\bar{\phi}_n) \,.
\label{psin}
\end{equation}  
\paragraph{Remarks}
\begin{itemize}
\item[(a)] The binary Darboux transformation is a powerful tool
to algebraically generate solutions to many integrable 
systems~\cite{MS91}. In particular, this method has been used
to investigate a broader class of rational solutions of KPI~\cite{ACTV00}.   
\item[(b)] It is possible to generate KP I lump solutions by
only using the ordinary Darboux transformation prescribed in
Lemma~\ref{lemma3}. This will lead to a wronskian form of
the $\tau$-function. However, the resulting solutions are in general, 
singular in $\mathbb{R}^2$ for any $t$.
\end{itemize}
Our next goal is to highlight certain interesting features of the 
eigenfunction $\psi_n$ and their role to characterize the $n$-lump NSE potential.
\subsection{Properties of the eigenfunction}
The function $\psi_0$ that was left arbitrary in \eqref{psin} is
now chosen to be $\psi_0 = A(k)\exp{i \theta}$ where
$\theta = kx+k^2y+k^3t+\theta_0(k)$,
$k \in \mathbb{C}$ an arbitrary complex parameter, and $A(k)$
to be determined a little later.
Inserting this $\psi_0$ into \eqref{psin} and using   
$\phi_n(k_0)= p_n(k_0)\exp{i\theta(k_0)}, \, k_0=a+ib$ together with
its complex conjugate $\bar{\phi}_n(k_0)$, one first computes
\[M(\psi_0,\bar{\phi}_n) = 
A(k)\sum_{j=0}^n \frac{i\bar{p}_{n-j}(k_0)}{(k-\bar{k}_0)^{j+1}} 
e^{i(\theta(k)-\bar{\theta}(k_0))}\,, \qquad
M(\phi_n,\bar{\phi}_n) = 
\frac{F_n}{2b}e^{i(\theta(k_0)-\bar{\theta}(k_0))} \,. \]
by using integration by parts and \eqref{taun},\eqref{Fn}.  
Substituting the above expressions back into \eqref{psin}
results in the following expression for the eigenfunction
\begin{equation}
\psi_n = \mu(k)e^{i \theta(k)}, \quad \qquad  
\mu(x,y,t,k) = 1+\frac{(k_0-\bar{k}_0)}{(k-k_0)}\left[
\frac{(F_n-p_n\bar{p}_n)}{F_n}
-\frac{p_n}{F_n}
\sum_{j=1}^n\frac{\bar{p}_{n-j}}{(k-\bar{k}_0)^j} \right] \,, 
\label{mu}
\end{equation}
where $A(k) = (k-\bar{k}_0)/(k-k_0)$ is chosen to ensure the normalization
$\mu(k) \to 1$ as $k \to \infty$. It is clear from \eqref{mu} that
the reduced eigenfunction $\mu(k)$ is meromorphic in the complex $k$-plane
with a simple pole at $k=k_0$ and a pole of order $n$ at $k=\bar{k}_0$.
In fact, by partial fraction decomposition, one can explicity write
\begin{gather}
\mu = 1 + \frac{\mu_1}{k-k_0}+
\sum_{j=1}^n\frac{\mu_{\bar{j}}}{(k-\bar{k}_0)^j}\,, \nonumber \\
\mu_1=(k_0-\bar{k}_0)\frac{F_n-p_n\bar{p}_n}{F_n}
-\frac{p_n}{F_n}\sum_{l=1}^n\frac{\bar{p}_{n-l}}{(k_0-\bar{k}_0)^{l-1}}\,,
\qquad \mu_{\bar{j}} = 
\frac{p_n}{F_n}\sum_{l=j}^n\frac{\bar{p}_{n-l}}{(k_0-\bar{k}_0)^{l-j}}\,,
\quad j=1,\ldots,n \,.  
\label{mupole}
\end{gather}
It is evident from \eqref{mupole} that $\mu(x,y,t,k)$ is a 
rational function in $x,y,t$. Moreover, 
for fixed $t$, one can deduce from \eqref{mupole}, \eqref{square} and \eqref{pn}
that $\mu_1 \sim O(\sf{1}{\theta_1}) \sim O(\sf{1}{R}), \, 
\mu_{\bar{j}} \sim O(\sf{1}{\bar{\theta}_1^j}) \sim O(\sf{1}{R^j})$
so that $\mu \to 1$ as $R=\sqrt{x^2+y^2} \to \infty$.

It follows from \eqref{lax} and \eqref{mu} that $\mu(k)$ satisfies the reduced NSE
\begin{equation}
i\mu_y = \mu_{xx}+2ik\mu_x+u_n\mu \,, 
\label{mueqn}
\end{equation}
as well as a time evolution equation which is not necessary to include here. 
Letting $k \to \infty$, one obtains from \eqref{mupole} and \eqref{mueqn} 
the following relation for the $n$-lump solution in terms of the residues
at the poles of $\mu(k)$
\[ u_n = -2i(\mu_1+\mu_{\bar{1}})_x \,,\]
which provides the NSE potential in terms of the eigenfunction. That is
usually the case in the IST method where one solves first for $\mu(k)$
and then obtains $u_n$ from it. However, in our case it is possible to
verify the above relation from the explicit expressions for $\mu(k)$ 
and $u_n$ that are already available. Since 
$\mu_1+\mu_{\bar{1}} = \sf{2ib}{F_n}(F_n-p_n\bar{p}_n)$ from \eqref{mupole}
and $u_n=2(\ln F_n)_{xx}$, it suffices to verify that 
$F_{nx}=2b(F_n-p_n\bar{p}_n)$. The latter follows immediately from
the expression of $F_n$ in \eqref{Fn} and noting that 
$\partial_x^{2n+1}|p_n|^2=0$.  

Near the pole $k=k_0$, the reduced eigenfunction in \eqref{mu} can
be separated into a singular and regular part $\mu(k) = \nu(k)+\mu_1/(k-k_0)$ 
where $\nu(k)$ is analytic near $k=k_0$ and $\nu(k) \to 1$ as $k \to \infty$.
Inserting this decomposition in \eqref{mueqn}, the coefficient of
$(k-k_0)^{-1}$ yields
\[L(k_0)\mu_1 =0, \quad \qquad 
L(k)=i\partial_y-\partial_{xx}-2ik\partial_x-u_n \,,\]
which implies that $\mu_1(k_0)$ is the eigenfunction of the linear operator
$L(k)$ with eigenvalue $k=k_0$. Integrating the equation $L(k_0)\mu_1 =0$ 
over $\mathbb{R}^2$ for fixed $t$ and using Green's theorem, one obtains
\[\int\!\!\int_{\mathbb{R}^2}u_n\mu_1 = 
\oint_{\Gamma_\infty}-i\mu_1(dx+2k_0dy)-\mu_{1x}dy =
-\oint_{\Gamma_\infty}\mu_1\, d\theta_1 \,, \]
where $\Gamma_\infty$ is a positively oriented 
contour at the boundary of $\mathbb{R}^2$ (i.e., $R \to \infty$)
surrounding the origin once. The last equality above follows
from \eqref{theta} and the fact that $\mu_{1x} \sim O(\sf{1}{R^2})$.
From \eqref{square} and \eqref{mupole}, one obtains as $R \to \infty$  
that $\mu_1 = -p_{n-1}/p_n + O(\sf{1}{R^2})=
-\partial_{\theta_1}\log p_n + O(\sf{1}{R^2})$ and from \eqref{pn}, 
$p_n \sim \theta_1^n/n!$. Thus,
\begin{equation*}
\frac{1}{2\pi i}\int\!\!\int_{\mathbb{R}^2}u_n\mu_1 = 
\frac{1}{2\pi i} \oint_{\Gamma_\infty} n\frac{d\theta_1}{\theta_1} = n\,. 
\end{equation*}
The above expression gives the winding number of the generalized Schur 
polynomial $p_n$ in the complex $\theta_1$-plane, and
defines a topological {\it charge}~\cite{AV97}, which provides
another characterization of the $n$-lump NSE
potential $u_n$. In the IST theory the charge is a constraint on the
NSE potential required to solve for $\mu(k)$ in \eqref{mu} in a 
consistent manner.  
In a similar way, writing 
$\mu(k)=\tilde{\nu}(k)+\mu_{\bar{1}}/(k-\bar{k}_0)
+ \ldots + \mu_{\bar{n}}/(k-\bar{k}_0)^n$, where $\tilde{\nu}(k)$ is analytic
near $k=\bar{k}_0$, the following chain of equations
\[L(\bar{k}_0)\mu_{\bar{j}} = 2i\partial_x\mu_{\overline{j+1}}\,, \quad j=1,2,\ldots,n-1,
\qquad \quad L(\bar{k}_0)\mu_{\bar{n}} = 0 \]
are obtained from the coefficients of $(k-\bar{k}_0)^j$ in \eqref{mueqn}.
In this case the linear operator $L(k)$ has the eigenfunction 
$\mu_{\bar{n}}(\bar{k}_0)$ with eigenvalue $k=\bar{k}_0$.  
It is possible to obtain a similar topological relation as above
if one integrates
the quantity $u_n\mu_{\bar{1}}$ over $\mathbb{R}^2$. Integrating
the equation $L(\bar{k}_0)\mu_{\bar{1}} = 2i\mu_{\bar{2}x}$
arising from the coefficient of $(k-\bar{k}_0)^{-1}$
and applying Green's theorem yields,
\[\int\!\!\int_{\mathbb{R}^2}u_n\mu_{\bar{1}} = 
\oint_{\Gamma_\infty}\mu_{\bar{1}}d\bar{\theta}_1
=n \oint_{\Gamma_\infty}\frac{d\bar{\theta}_1}{\bar{\theta}_1} 
=-2n\pi i \,, \]
after using the fact that from \eqref{mupole} 
$\mu_{\bar{1}x} \sim O(\sf{1}{R^2}), \, \mu_{\bar{2}} \sim O(\sf{1}{R^2})$
as $R \to \infty$ (so they do not contribute to the line integral), and 
$\mu_{\bar{1}} = \bar{p}_{n-1}/\bar{p}_n + O(\sf{1}{R^2})
=n/\theta_{\bar{1}}+O(\sf{1}{R^2})$.     
Combining the two results above, one can give a unified definition
of the topological charge as~\cite{VA99}
\begin{equation} 
Q(k_{\alpha}) = \frac{\mathrm{sgn}(\Im(k_{\alpha}))}{2 \pi i}
\int\!\!\int_{\mathbb{R}^2}u_n\mathrm{Res}_{k_\alpha}\mu(k) = n \,,
\qquad k_\alpha \in \{k_0, \bar{k}_0\} \,. 
\label{Q}
\end{equation}
In fact $Q(k_0)=Q(\bar{k}_0)$ can also be shown by integrating
$u_n(\mu_1+\mu_{\bar{1}})$ over $\mathbb{R}^2$ and using the residue
relation $u_n=-2i(\mu_1+\mu_{\bar{1}})_x$.   
For the sake of completeness we mention that using \eqref{mupole}
and Green's theorem, it can also be shown that
\[\int\!\!\int_{\mathbb{R}^2}u_n\mu_{\bar{j}} =0, \quad 2 \leq j \leq n\,, \]
since $\mu_{\bar{j}} \sim O(\sf{1}{R^j})$ as $R \to \infty$.
Recall that $u_n \notin L^1(\mathbb{R}^2)$ since it decays as $\sf{1}{R^2}$.
The slow decay is the underlying fact why the charge appears.
\paragraph{Remarks}
\begin{itemize}
\item[(a)] Our discussion in Section 5.2 folows
closely the spectral theory of KPI rational solution developed
in \cite{AV97,VA99} although our approach is algebraic.  
The explicit expression for $\mu(k)$ in \eqref{mupole} was
not presented earlier since the resulting computation via IST
becomes rather unwieldy. 
\item[(b)] The IST approach in~\cite{AV97,VA99} assumes that the
kernel of the operator $L(k_0)$ is one-dimensional. That is
known {\it not} to be true for more general multi-lump solutions of KPI
although it has been conjectured~\cite{AV04} that 
$\dim(\ker(L(k_0))) = Q+2-(m+\bar{m})$ where $Q$ is the charge and 
$m,\bar{m}$ are the order of the
poles of $\mu(k)$ at $k_0$ and $\bar{k}_0$, respectively.
It should be interesting to derive the above relation for 
$\dim(\ker(L(k_0)))$ via the direct, algebraic
approach where $\mu(k)$ is explicitly constructed
by the binary Darboux transformation.
\end{itemize}

\section{Concluding remarks}
The rational multi-lump solutions of KPI equation had been
found earlier but their internal dynamics have not been
studied in great details because the underlying formulas are
complicated. For this reason, we have considered a special class of
solutions that are relatively easier to analyze as well as
retain the flavor of anomalous interactions exhibited by such types
of rational solutions. We have comprehensively studied the behavior
and properties of this class of solutions in this paper.
In the future, we plan to investigate the general type of multi-lump
solutions of KPI, characterize them and develop a suitable
classification scheme. We would also like to relate our
work to the IST scheme in the hope of addressing certain
open issues. 

\section{Acknowledgments}
SC thanks Prof. Mark Ablowitz (CU Boulder) for useful discussions. 
Early part of this work was partially supported
by NSF grant No. DMS-1410862.


\section*{Appendix}
\appendix
\section{Approximate vs exact peak locations}
Throughout Sections 3 and 4, the approximation
$F_n^{(0)}(r,s,t)=0$ equivalently,
$\tilde{p}_n(r,s,t)=0$ was used to obtain the locations of the lump peaks
instead of their exact values that are obtained from maximizing $u_n(r,s,t)$.
As mentioned in Section 3.2, this approximation is based the assumption
that the maxima of $u_n$ are close to the minima of $F_n$ which again,
are close to the zeros of the highest degree polynomial $F_n^{(0)}$
in the expression \eqref{square} for $F_n$. In fact, the exact and approximate
peak locations were shown in Section 3.2 to be off by $O(|t|^{-1/2})$ for $|t| \gg 0$
in the case of $2$-lump solutions. Here we are going to provide a
justification that this
approximation is also valid for the $n$-lump solution
by utilizing the asymptotics developed in Sections 4.2 and 4.3.
In the following we will denote the exact peak location by
$z_p:=(r_p,s_p)$, the approximate peak location via minimizing $F_n$
by $z_p^{\mathrm min} = (r_p^{\mathrm min}, s_p^{\mathrm min})$, and
the approximate peak location from $F_n^{(0)}=0$ by
$z_p^{(0)} = (r_p^{(0)}, s_p^{(0)})$. Then it will be shown specifically that
\begin{equation}
z_p \sim z_p^{\mathrm min} + O(|t|^{-1/2}), \qquad
z_p^{\mathrm min} \sim z_p^{(0)} + O(|t|^{-1/2}), \qquad |t| \gg 0 \,.
\label{zpeak}
\end{equation}
We will give the details of the second
asymptotic calculation below and then will briefly outline the first one.
 
Differentiating $F_n = F_n^{(0)}+l = A^2+B^2+l$ with respect to $r$ and $s$
one readily obtains expressions for $F_{nr} := g, \, F_{ns} := h$ and
their derivatives $g_r,\, h_s,\, g_s=h_r$. Setting $A=B=0$
into these expressions, yields the following at $z=z_p^{(0)}$
\begin{gather*}
g = l_r,  \qquad h=l_s, \qquad g_r=2(A_r^2+B_r^2)+l_{rr}, \\
h_s = 2(A_s^2+B_s^2)+l_{ss}, \qquad  g_s=h_r=2(A_rA_s+B_rB_s)+l_{rs} \,.  
\end{gather*}
If we linearize $g$ and $h$ near $z=z_p^{(0)}$ and evaluate at
$z=z_p^{\mathrm min}$ where $g=h=0$, we obtain the following approximate
equations
\begin{align}
0 = g(z_p^{\mathrm min}) & \approx  g(z_p^{(0)}) + \delta r\, g_r(z_p^{(0)}) 
+ \delta s\, g_s(z_p^{(0)}) \nonumber \\
0 = h(z_p^{\mathrm min}) & \approx  h(z_p^{(0)}) + \delta r\, h_r(z_p^{(0)}) 
+ \delta s\, h_s(z_p^{(0)}) \,,
\label{gh}
\end{align}
which can be regarded as a linear system $M\delta u = v$ for 
$\delta u = (\delta r, \, \delta s)^T$,
$v= - (l_r, \, l_s)(z_p^{(0)})$. We next proceed to obtain asymptotic
expressions for the entries of the coefficient matrix $M$ which is real symmetric.
It is convenient first to express these entries in terms of generalized
Schur polynomials. It is known from \eqref{lasymp} that
$l(z_p^{(0)}) \sim |\hat{p}_{n-1}|^2 \sim O(|t|^{n-1})$ for $|t| \gg 0$.
Consequently, $l_r,l_s \sim O(|t|^{n-3/2})$ and $l_{rr},l_{rs},l_{ss} \sim O(|t|^{n-2})$
at $z=z_p^{(0)}$. (Recall that $z_p^{(0)}$ represents any of the approximate
peak locations that were denoted by $z_j(t), \, j=1,\ldots,n$ in Sections 4.1 and 4.2).
Since $A+iB=\tilde{p}_n$, one can express the quantity 
$A_r^2+B_r^2 = |\partial_r(A+iB)|^2 = |\tilde{p}_{n-1}|^2$,
and similarly,
$A_s^2+B_s^2~= |\partial_s(A+iB)|^2 = |\tilde{p}_{n-1}-\sf{i}{2b}\tilde{p}_{n-2}|^2$
where \eqref{thetars} and \eqref{propa} were used to compute the derivatives.
Moreover, $(A_r+iB_r)(A_s-iB_s) = (A_rA_s+B_rB_s)-i(A_rB_s-A_sB_r)
= -i|\tilde{p}_{n-1}|^2+\sf{1}{2b}\tilde{p}_{n-1}\overline{\tilde{p}}_{n-2}$
where the overbar indicates complex conjugation. Equating the real and imaginary part
of the previous expression give
$A_rA_s+B_rB_s = \sf{1}{2b}\Re(\tilde{p}_{n-1}\overline{\tilde{p}}_{n-2})$
and $A_rB_s-A_sB_r=|\tilde{p}_{n-1}|^2-
\sf{1}{2b}\Im(\tilde{p}_{n-1}\overline{\tilde{p}}_{n-2})$. Finally, using
\eqref{pndominant-a} and \eqref{pndominant-b}, the leading order behavior
of the entries of the matrix $M$ for $|t| \gg 0$ are given as follows
\begin{gather}
g_r \sim 2(A_r^2+B_r^2) = 2|\tilde{p}_{n-1}|^2 \sim O(|t|^{n-1}), \nonumber \\
h_s \sim 2(A_s^2+B_s^2) = 2|\tilde{p}_{n-1}-\sf{i}{2b}\tilde{p}_{n-2}|^2 
\sim 2|\tilde{p}_{n-1}|^2 \sim O(|t|^{n-1}) \label{ghasymp} \\
g_s=h_r \sim 2(A_rA_s+B_rB_s) = 
\sf{1}{2b}\Re(\tilde{p}_{n-1}\overline{\tilde{p}}_{n-2}) \sim O(|t|^{n-3/2}) \,,
\nonumber
\end{gather}
at $z=z_p^{(0)}$. It follows from \eqref{ghasymp}
that the off diagonal entries of $M$ in \eqref{gh}
are of lower order in $|t|$ than the diagonal entries. The determinant of $M$ to
leading order in $|t|$ is given by
\[\det M = g_rh_s-g_sh_r \sim 4(A_rB_s-A_sB_r)^2 = 4 \big(|\tilde{p}_{n-1}|^2-
\sf{1}{2b}\Im(\tilde{p}_{n-1}\overline{\tilde{p}}_{n-2})\big)^2
\sim 4|\tilde{p}_{n-1}|^4 \sim O(|t|^{2(n-1)}) \,,\]
which shows that $\det M > 0$, hence $M$ is invertible. Then
from the solution $\delta u = (\delta r, \, \delta s)^T=M^{-1}v$ of the
linear system \eqref{gh}  
one can estimate $\|\delta u\| \leq \|M^{-1}\|\|v\|$ in any typical matrix norm.
Since $l_r,l_s \sim O(|t|^{n-3/2})$, it is easily seen that
$\|v\| \sim O(|t|^{n-3/2})$ and from the leading order behaviors 
given in \eqref{ghasymp}, $\|M^{-1}\| \sim O(|t|^{-(n-1)})$,
which implies that $\|\delta u\|$ is at most $O(|t|^{-1/2})$. Thus the approximate
peak locations $z_p^{(0)}$ and $z_p^{\mathrm min}$ are within $O(|t|^{-1/2})$
for large $|t|$, which establishes the second assertion in \eqref{zpeak}.

In order to establish the first assertion in \eqref{zpeak},
we proceed in a similar way as above. Namely, we evaluate the partial 
derivatives $G := u_{nr}$
and $H := u_{ns}$ at
$z_p^{\mathrm min}$ where $F_{nr}=F_{ns}=0$. Indeed  setting
$z=z_p$ in the linear approximation of $G(z), \, H(z)$ near
$z=z_p^{\mathrm min}$ and using the fact that $G(z_p)=H(z_p)=0$,
it is possible
to obtain a linear system for $\delta u = (\delta r, \, \delta s)^T$ with
$v=(G,\, H)(z_p^{\mathrm min})$ and the coefficient matrix $M$ which
has entries $G_r,H_s,G_s=H_r$ evaluated at $z=z_p^{\mathrm min}$.
In order to find the dominant behavior of the entries of $M$ for
$|t| \gg 0$ we actually evaluate them at $z=z_p^{(0)}$ instead of $z_p^{\mathrm min}$.
The error due to this approximation only gives a lower order (in $|t|$)
correction term to the leading order estimate $\delta u \sim O(|t|^{-1/2})$
obtained this way. We do not include the detailed asymptotic expressions
since they are straightforward but cumbersome.

\end{document}